\documentclass[journal]{IEEEtran}
\usepackage{graphicx}
\usepackage{color}

\usepackage{pb-diagram}
\usepackage[cmex10,centertags]{amsmath}
\usepackage{amsfonts}
\usepackage{amssymb}

\usepackage{cite}

\usepackage[utf8]{inputenc}
\usepackage[T1]{fontenc}

\usepackage{url}

\def\C{{\mathbb C}}

\def\R{{\mathbb R}}

\def\Z{{\mathbb Z}}
\def \P{\mathbb P}

\def\CV{\mathcal{V}}

\def\ra{{\rightarrow}}

\newcommand{\G}[2]{\mathbb{G}_{#1,#2}}
\DeclareMathOperator*{\myvec}{vec}
\DeclareMathOperator*{\rank}{rank}

\newcommand{\bi}{\begin{itemize}}
\newcommand{\ei}{\end{itemize}}
\newcommand{\bd}{\begin{description}}
\newcommand{\ed}{\end{description}}
\newcommand{\beq}{\begin{equation}}
\newcommand{\eeq}{\end{equation}}
\newcommand{\beqn}{\begin{eqnarray}}
\newcommand{\eeqn}{\end{eqnarray}}
\newcommand{\beqna}{\begin{eqnarray*}}
\newcommand{\eeqna}{\end{eqnarray*}}

\newtheorem{corollary}{Corollary}
\newtheorem{lemma}{Lemma}
\newtheorem{theorem}{Theorem}
\newtheorem{proposition}{Proposition}
\newtheorem{conjecture}{Conjecture}%[section]

\newtheorem{problem}{Problem}
\newtheorem{example_scenario}{Example}

\newtheorem{remark}{Remark}

\begin{document}
%\bibliographystyle{unsrt}
\title{A Feasibility Test for Linear Interference Alignment in MIMO Channels with Constant Coefficients}
\author{\'Oscar~Gonz\'alez, \IEEEmembership{Student~Member,~IEEE}, Carlos~Beltr\'an, and~Ignacio~Santamar\'ia,~\IEEEmembership{Senior~Member,~IEEE}\thanks{\'O. Gonz\'alez and I. Santamar\'ia are with the Departamento de Ingenier\'ia de Comunicaciones (DICOM), Universidad de Cantabria, Santander, 39005, Spain. C. Beltr\'an is with the Departamento de Matem\'aticas, Estad\'istica y Computaci\'on, Universidad de Cantabria. Avda. Los Castros s/n, Santander, Spain. The work of \'O. Gonz\'alez and I. Santamar\'ia was supported by MICINN (Spanish Ministry for Science and Innovation) under grants TEC2010-19545-C04-03 (COSIMA), CONSOLIDER-INGENIO 2010 CSD2008-00010 (COMONSENS) and FPU grant AP2009-1105. Carlos Beltr\'an was partially supported by the MICINN grant MTM2010-16051.

This paper was presented in part at the IEEE 2012 International Symposium on Information Theory (ISIT 2012), Cambridge, MA, USA.

\noindent Copyright (c) 2013 IEEE. Personal use of this material is permitted.  However, permission to use this material for any other purposes must be obtained from the IEEE by sending a request to pubs-permissions@ieee.org.
}
}
%\thanks{\scriptsize{}}}%

%% make the title area
\maketitle
% The abstract should be less than 200 words, self-contained (no citations), written in the passive voice

\begin{abstract}
In this paper, we consider the feasibility of linear interference alignment (IA) for multiple-input multiple-output (MIMO) channels with constant coefficients for any number of users, antennas and streams per user; and propose a polynomial-time test for this problem. Combining algebraic geometry techniques with differential topology ones, we first prove a result that generalizes those previously published on this topic. Specifically, we consider the input set (complex projective space of MIMO interference channels), the output set (precoder and decoder Grassmannians) and the solution set (channels, decoders and precoders satisfying the IA polynomial equations), not only as algebraic sets but also as smooth compact manifolds. Using this mathematical framework, we prove that the linear alignment problem is feasible when the algebraic dimension of the solution variety is larger than or equal to the dimension of the input space {\it and} the linear mapping between the tangent spaces of both smooth manifolds given by the first projection is generically surjective. If that mapping  is not surjective, then the solution variety projects into the input space in a {\em singular way} and the projection is a zero-measure set. This result naturally yields a simple feasibility test, which amounts to checking the rank of a matrix. We also provide an exact arithmetic version of the test, which proves that testing the feasibility of IA for generic MIMO channels belongs to the bounded-error probabilistic polynomial (BPP) complexity class.
\end{abstract}

%\IEEEpeerreviewmaketitle

%-----------------------------------------------------------------

%-----------------------------------------------------------------
\section{Introduction}

\IEEEPARstart{T}{he} degrees of freedom (DoF) of a wireless interference network represent the number of non-interfering data streams that can be simultaneously transmitted over the network. Recently, it has been shown that to achieve the maximum DoF of the $K$-user multiple-input multiple-output (MIMO) interference channel, the interference from other transmitters must be aligned at each receiver in a lower-dimensional subspace \cite{JafarTut}. This is the basic idea of the interference alignment (IA) technique which first
originated out of the study of the degrees of freedom of the 2-user X channel  \cite{JafarX08,Motahari08}, shortly afterwards was extended to the $K$-user interference channel \cite{Jafar08}, and has received a lot of attention since then.

In this paper we consider the alignment problem for the $K$-user MIMO interference channel with constant channel coefficients. Also, we restrict our attention to IA schemes that apply linear decoders and precoders without channel or symbol extensions, which means that the MIMO channel matrices have no particular structure (e.g., diagonal or block diagonal)\footnote{We do not consider in this paper interference alignment schemes on the signal scale which are based on the properties of rationals and irrationals \cite{Ghasemi09},\cite{Etkin09},\cite{Wu11}.}. For this setting, when all transmitters and receivers have the same number of antennas, the ratio of total DoF to the single user DoF is upper bounded by 2, whereas this ratio increases to $K/2$ for frequency- or time-varying channels when the channel extensions are i.i.d. and exponentially long in $K$ \cite{Jafar08},\cite{Yetis10}. However, requiring channels to have an unbounded number of extensions can be a limiting factor in practice and, consequently, alignment in signal space with constant MIMO interference channels has been the preferred option for recent experimental studies on IA \cite{Katabi09},\cite{Ayach10},\cite{Oscar11}.

In this paper, we address the feasibility of linear IA for MIMO interference networks with constant channel coefficients and no symbol extensions. Our focus is the single channel use IA feasibility problem, which has recently received a lot of attention, and results herein do not apply if multiple channel uses are considered. This problem amounts to solving a set of polynomial equations and some partial results can be found in \cite{Yetis10},\cite{BreslerCartwrightTseToappear},\cite{Razaviyayn1}. The first work to study this problem was \cite{Yetis10}, where the solvability of the IA polynomial equations was analyzed using classic results in algebraic geometry like Bezout's and Bernstein's theorems. By counting the number of equations and variables involved in any subset of zero-forcing alignment equations, Yetis et al. introduced in \cite{Yetis10} the definition of a proper system. Connections between proper and feasible systems were established only for the single-beam case in which each user transmits only one stream of data. When more than one data stream is transmitted, the genericity of the polynomial coefficients is destroyed and the equivalence between proper and feasible systems does not longer hold. Some information theoretic outer bounds, e.g., \cite{JafarDoF07} and \cite{JafarDoF10}, can be included in the properness definition to further close the gap between proper and feasible systems, but the precise connection between both concepts still remains unclear.

In \cite{Slock10}, the feasibility of IA was studied by interpreting the alignment process as a joint transmit-receive zero-forcing scheme in which each interfering stream can be suppressed at either the transmitter or the receiver sacrificing one degree of freedom. The proposed feasibility test, however, provides only necessary conditions and is combinatorial in nature since it requires to check all possible ways to suppress interfering streams at both sides of the link and for all users.

More recent work on the feasibility of IA has been presented in \cite{BreslerCartwrightTseToappear} and \cite{Razaviyayn1}. Specifically, \cite{BreslerCartwrightTseToappear} studies the dimensions of the algebraic varieties involved in the alignment problem (input, output and solution variety), and proves a sufficient and necessary condition of feasibility for the particular case of symmetric square MIMO interference channels, where all transmitters and receivers have the same number of antennas, all users transmit the same number of streams and there are at least three interfering users ($K \geq 3$). For the general case with arbitrary system parameters, only a necessary condition is proved in \cite{BreslerCartwrightTseToappear}. Similar algebraic tools are used in \cite{Razaviyayn1} to prove general bounds on the tuple of DoF that are achievable through linear interference alignment. Furthermore, for the particular case of symmetric systems where the number of transmit and receive antennas is divisible by the number of streams the bound is tight and can be achieved through IA.

In this work, we first prove a slight generalization of the results in \cite{BreslerCartwrightTseToappear} and \cite{Razaviyayn1} that fully characterizes the feasibility of linear interference alignment for MIMO channels with constant coefficients and no symbols extension in arbitrary settings (for any number of users, antennas and streams per user, and non necessarily fully-connected networks). To derive this result, we combine algebraic geometry techniques with differential topology ones and consider the three sets involved in the problem (i.e., the input set formed by the Cartesian product of complex projective matrices, $\mathcal{H}$, the output set formed by the Cartesian product of precoder and decoder Grassmannians, $\mathcal{S}$, and the solution variety formed by tuples of channels, decoders and precoders satisfying the alignment equations, $\CV$), not only as algebraic sets but also as smooth compact manifolds. Viewing the channels, the decoders and precoders, and the solution variety as compact manifolds, some important results stand out from the study of their tangent spaces. In words, we prove that the linear alignment problem is feasible when the algebraic dimension of $\CV$ is larger than or equal to the dimension of $\mathcal{H}$ {\it and} the linear mapping between the tangent spaces of both smooth manifolds given by the first projection is  surjective. If the mapping between the tangent spaces of $\CV$ and $\mathcal{H}$ is not surjective, then the whole set $\CV$ projects into $\mathcal{H}$ in a {\em singular way} and the projection is a zero-measure set. This situation explains those systems that are proper, but infeasible.

This result enables us to derive the main contribution of this paper, which is a simple feasibility test that amounts to checking the rank of a certain matrix. We provide floating-point and exact arithmetic versions of the test, as well as a detailed complexity analysis which proves that the problem of deciding infeasibility for generic MIMO channels belongs to the bounded-error probabilistic polynomial-time (BPP) complexity class in the Turing Machine model of computation. Using the proposed test we were able to study the feasibility of systems with a large number of antennas and users and, from the general trends observed, to put forward a conjecture on the number of linear DoF of symmetric $M \times N$ MIMO interference channels. Also, the proposed feasibility test can also be used to obtain the total DoF for any arbitrary K-user MIMO interference channel without resorting to the existing inner and outer information-theoretic bounds. Some work along this line has recently been discussed in \cite{Gonzalez2013}.

The paper is organized as follows. In Section \ref{systemmodel}, the system model is introduced and the IA feasibility problem is formally stated. In Section \ref{main} we present a result that characterizes the feasibility of linear interference alignment for MIMO channels with constant coefficients in arbitrary settings. The proposed feasibility test, which essentially consists of checking whether a certain matrix is rank-deficient or not, is presented in Section \ref{test}. In this section we also present floating-point and exact arithmetic versions of the test, and prove that the later describes a BPP Turing machine. In Section \ref{sec:proofs}, we prove the main theorems of the paper. In Section \ref{sec:simulations}, we validate our feasibility test in several symmetric and asymmetric interference channels showing that its results are consistent with other previously known results. Additionally, we use our test to establish a conjecture on the DoF of the $K$-user symmetric interference channel. Finally, the main conclusions are summarized in Section \ref{sec:conclusions}.

%\subsection{Notation}

%--------------------------------------------------------------------
\section{System model and problem statement}
\label{systemmodel}

\subsection{System model}

We consider in this paper the $K$-user MIMO interference channel with transmitter $j$ having $M_j\geq1$ antennas and receiver $j$ having $N_j\geq1$ antennas. Each user $j$ wishes to send $d_j\geq0$ streams or messages. We adhere to the notation used in \cite{Yetis10} and denote this asymmetric interference channel as $\prod_{k=1}^K \left(M_k\times N_k,d_k\right)=\left(M_1\times N_1,d_1\right)\cdots \left(M_K\times N_K,d_K\right)$. The symmetric case in which each user transmits $d$ streams and is equipped with $M$ transmit and $N$ receive antennas is denoted as $\left(M\times N,d\right)^K$. In the square symmetric case all users have the same number of antennas $M=N$.

The MIMO channel from transmitter $l$ to receiver $k$ is denoted as $H_{kl}$ and assumed to be flat-fading and constant over time. Each $H_{kl}$ is an $N_k\times M_l$ complex matrix (i.e., $H_{kl} \in \mathbb{C}^{N_k\times M_l}$). All channels are independent of each other and their entries are also independently drawn from a continuous distribution (channels generated in this way are said to be generic). We let $\Phi\subseteq\{1,\ldots, K\}\times\{1,\ldots,K\}$ be the (nonempty) subset of indexes $(k,l),k\neq l$ such that $H_{kl}$ is nonzero, therefore we assume that $H_{kl}$ is defined for $(k,l)\in\Phi$. Note that if $\Phi=\{(k,l):k\neq l\}$, then the interference channel is fully connected, otherwise the channel is partially connected, which can be due to path loss or shadowing \cite{Lau11}. Both scenarios are covered by the results in this paper. We will denote by $\sharp(\Phi)$ the number of elements in the (finite) set $\Phi$ (i.e., the non-zero interference links).

User $j$ encodes its message using an $M_j \times d_j$ precoding matrix $V_j$ and the received signal is given by
\begin{equation}
\label{eq:received}
y_j=H_{jj}V_jx_j + \sum_{i\neq j}H_{ji}V_ix_i +n_j, \hspace{1cm} 1\leq j \leq K
\end{equation}
where $x_j$ is the $d_j \times 1$ transmitted signal and $n_j$ is the zero mean unit variance circularly symmetric additive white Gaussian noise vector. The first term in (\ref{eq:received}) is the desired signal, while the second term represents the interference space. The receiver $j$ applies a linear decoder $U_j$ of dimensions $N_j \times d_j$, i.e.,
\begin{equation}
\label{eq:received1}
U_j^T y_j=U_j^T H_{jj}V_jx_j + \sum_{i\neq j}U_j^T H_{ji}V_ix_i +U_j^T n_j, \quad 1\leq j \leq K,
\end{equation}
where superscript $T$ denotes transpose.

%--------------------------------------------------------
\subsection{Problem statement}
\label{sec:problem_statement}
The interference alignment (IA) problem consists in finding the decoders and precoders, $V_j$ and $U_j$, in such a way that the interfering signals at each receiver fall into a reduced-dimensional subspace and the receivers can then extract the projection of the desired signal that lies in the interference-free subspace. To this end it is required that the polynomial equations
\begin{equation}\label{eq:1}
U_k^TH_{kl}V_l=0,\qquad (k,l)\in\Phi,
%Hermitian
\end{equation}
are satisfied, while the signal subspace for each user must be linearly independent of the interference subspace and must have dimension $d_k$, that is
\begin{equation}\label{eq:rank}
\rank(U_k^TH_{kk}V_k)=d_k,\qquad\forall\;k\in\{1,\ldots,K\}.
%Hermitian
\end{equation}
We recall that all matrices $H_{kl}$ (including direct link matrices, $H_{kk}$) are generic, that is, their entries are drawn from a continuous probability distribution and are independent of each other (independence among different links also holds). Consequently, \eqref{eq:rank} is satisfied almost surely. Thus, we will consider that solving the linear IA feasibility problem amounts to solve the polynomial equations in \eqref{eq:1} only.

In this paper we are interested in studying the relationship between $d_j,M_j,N_j,K$ such that the linear alignment problem is feasible. For example, we may want to know: for given $K$ and $d_j$, which collections of $M_j,N_j$ make the problem feasible (for every possible choice of the matrices $H_{kl}$), or for given $K$ and $M_j,N_j$, which are the greatest values for $d_j$ that can be achieved? In the later case, the tuple $\left( d_1,\ldots,d_K \right)$ defines the maximum degrees of freedom (DoF) of the system, that is the maximum number of independent data streams that can be transmitted without interference in the channel.

%%%%%%%%%%%%%%%%%%%%%%%%5
It is well-known that the number of streams transmitted by all users must satisfy the point-to-point bounds
\begin{equation}\label{eq:p2p_assumption1_fullyconnected}
1\leq d_j\leq\min(N_j,M_j),\qquad\forall\;j\in\{1,\ldots,K\}.
\end{equation}
Note that we can exclude the case that some $d_j=0$ without loosing generality, because it amounts to removing all pairs containing the index $j$ from $\Phi$. From a mathematical point of view, in the general (not necessarily fully connected) case, the natural substitute of (\ref{eq:p2p_assumption1_fullyconnected}) is the following:
\begin{equation}\label{eq:p2p_assumption1_prev}
1\leq d_k\leq N_k,\qquad  1\leq d_l\leq M_l,\qquad\forall (k,l)\in\Phi.
\end{equation}
We want to state absolutely general results, which leads us to consider the two following sets:
\[
\Phi_R=\{k\in\{1,\ldots,K\}:\,\exists\,l\in\{1,\ldots,K\},(k,l)\in\Phi\},
\]
\[
\Phi_T=\{l\in\{1,\ldots,K\}:\,\exists\,k\in\{1,\ldots,K\},(k,l)\in\Phi\}.
\]
Note that $\Phi_R$ ($\Phi_T$) is the first (second) projection of the set $\Phi$. In words, $\Phi_R$ indicates the set of receivers which suffer interference from at least one transmitter, whereas $\Phi_T$ contains the set of transmitters which provoke interference to at least one receiver. Then, (\ref{eq:p2p_assumption1_prev}) is equivalent to
\begin{equation}\label{eq:p2p_assumption1}
1\leq d_k\leq N_k,\; \hspace{0.1cm}  \forall\;k\in\Phi_R,\qquad 1\leq d_l\leq M_l,\; \hspace{0.1cm}\forall\;l\in\Phi_T.
\end{equation}
Equations (\ref{eq:p2p_assumption1_fullyconnected}) and (\ref{eq:p2p_assumption1}) are equivalent if each user interferes at least one user and it is interfered by at least one user, that is if $\Phi_R=\Phi_T=\{1,\ldots,K\}$. In particular, they are equivalent in the fully--connected case. Note also that if $l\not\in\Phi_T$ then the precoder $V_l$ does not appear in the equations (\ref{eq:1}) and plays no role in the problem, thus it consists of free variables. We deem that it is more appropriate not to consider these free variables as part of the problem. Hence, if for example we say that the problem has finitely many solutions we mean that the number of solutions of the non-free variables is finite (although, if there is some $l\not\in\Phi_T$, there will be infinitely many ways to choose $V_l$). The same can be said if $k\not\in\Phi_R$ for some $k$.

Additionally, note that if user $l$ transmits all possible streams according to its point-to-point bound, $d_l=M_l$ (which implies that $M_l \leq N_l$); then, it is not possible for user $k\neq l$, with $(k,l)\in \Phi$, to also reach its point-to-point bound with equality and thus receive $d_k=N_k$ desired streams (with $N_k \leq M_k$). This stems from the fact that receiver $k$ has to leave at least a one-dimensional subspace for the interference, otherwise the desired signal subspace would not be free of interference. In other words, the two users of an interference link cannot reach their point-to-point bounds simultaneously. Formally, this condition can be stated as the following set of necessary conditions
\begin{equation}\label{eq:p2p_assumption2}
N_kM_l>d_kd_l, \qquad\forall (k,l)\in\Phi,
\end{equation}
which complement the direct link conditions in \eqref{eq:p2p_assumption1}. To derive our results we only assume that both \eqref{eq:p2p_assumption1} and \eqref{eq:p2p_assumption2} hold.

There are other necessary conditions for feasibility that involve two or more users. Specifically, in \cite{JafarDoF07} it was proved that for the $2$-user MIMO interference channel consisting of users $k$ and $l$, if $(k,l)\in\Phi$ and $(l,k)\in\Phi$, the DoF satisfy
\begin{align}\label{eq:DoF}
d_k+d_l\leq \min &\left( M_l+M_k, N_l+N_k,\right. \notag \\ &\left.\phantom{(}\max(N_k,M_l), \max(N_l,M_k) \right).
\end{align}
For the symmetric $K$-user MIMO interference channel\footnote{Let us remind again that we are only considering the DoF achievable with linear alignment schemes and without symbol extensions. When lattice-based alignment schemes are used, the DoF of interference channels with real and constant coefficients have been studied in \cite{Ghasemi09},\cite{Etkin09},\cite{Wu11}.}, the following outer bound for the total number of DoF was proved in \cite{JafarDoF10}
\begin{align}\label{eq:DoF_K}
d_1+\cdots +d_K &\leq K \min(M,N)\,\emph{I}\left( K \leq R\right)\notag \\
&+ K \frac{\max \left (M,N \right)}{R+1}\emph{I}\left( K > R \right),
\end{align}
where $\emph{I} \left( \cdot \right)$ represents the indicator function and $R=\lfloor \max \left (M,N \right)/\min \left (M,N \right)\rfloor$.

%%%%%%%%%%%%%%%%%%%%%%%%%%%

Our techniques for proving the main results will come from algebraic geometry and differential topology. Our arguments are sometimes similar to those in \cite{BreslerCartwrightTseToappear},\cite{Razaviyayn1}, with the difference that not only the algebraic nature of the objects is used, but also their smooth manifold structures, as well as the key property of compactness. We are greatly inspired by Shub and Smale's construction for polynomial system solving, see \cite{ShSm93b} or \cite{BlCuShSm98}. Some basic knowledge of smooth manifolds is assumed. More advanced results on differential topology that will also be used during the derivations are relegated to Appendix \ref{app:difftop}.

To formally state the IA feasibility problem, it is convenient to first define three tuples: $H$, $U$ and $V$. $H$ denotes the collection of all $H_{kl},\,(k,l)\in\Phi$ and, similarly, $U$ and $V$ denote the collection of $U_k$, $k\in\Phi_R$ and $V_l$, $l\in\Phi_T$, respectively. Even though for the system model described in \eqref{eq:received} and \eqref{eq:received1} we have used the symbols $H_{kl}$, $U_k$ and $V_l$ for complex matrices, in the following we will show that to solve the problem is more convenient to let them live in different spaces that take into account the invariances of (\ref{eq:1}). If $(H,U,V)$ satisfies (\ref{eq:1}) then we can multiply each matrix $H_{kl}$ in $H$ by a nonzero complex number and (\ref{eq:1}) will still hold. Thus, it makes sense to consider our matrices as elements of the projective space of matrices, i.e., we can think of $H_{kl}$ as a whole line in $\mathbb{C}^{N_k\times M_l}$. Similarly, we can think of each $U_k$ (equiv. $V_l$) as a subspace spanned by the columns of a $N_k\times d_k$ (equiv. $M_l\times d_l$) matrix. From a mathematical point of view, this consideration permits us to use projective spaces and Grassmannians (which are both compact spaces) instead of non-compact affine spaces.

Thus, we consider the projective space of complex channel matrices, $\P(\mathbb{C}^{N_k\times M_l})$, and the Grassmannians\footnote{For integers $1\leq a\leq b$ we denote as $\G{a}{b}$ the Grassmannian formed by the linear subspaces of (complex) dimension $a$ in $\C^b$.} formed by the decoders and precoders. With some abuse of notation we will refer to their elements as $H_{kl}$ and $U_k$, $V_l$, respectively. More formally, given $\sharp(\Phi)$ elements
\[
H_{kl}\in\P(\mathbb{C}^{N_k\times M_l}),\quad (k,l)\in\Phi,
\]
to solve the IA problem one would like to find a collection of subspaces
\[
U_k\in\G{d_k}{N_k},\quad k\in\Phi_R,\qquad V_l\in\G{d_l}{M_l},\quad l\in\Phi_T
\]
such that the polynomial equations (\ref{eq:1}) are satisfied. The (generic) IA feasibility problem consists on deciding whether, given $K,M_j,N_j,d_j$ and $\Phi$,  all or almost all choices of $H_{kl}$ will admit such $U_k,V_l$.
%\begin{equation}\label{eq:1}
%U_k^TH_{kl}V_l=0,\qquad (k,l)\in\Phi,
%\end{equation}
We have already pointed out that the IA equations given by (\ref{eq:1}) hold or do not hold independently of the particular chosen affine representatives of $(H,U,V)$.

As in \cite{BreslerCartwrightTseToappear}, the proof of our main theorems will follow from the study of the set $\{(H,U,V): \text{ (\ref{eq:1}) holds}\}$. More precisely, consider the following diagram
\begin{equation}\label{eq:diag}
\begin{matrix}
&&\CV&&\\
\pi_1&\swarrow&&\searrow&\pi_2\\
\mathcal{H}&&&&\mathcal{S}
\end{matrix}
\end{equation}
where
\[
\mathcal{H}=\prod_{(k,l)\in\Phi}\P(\mathbb{C}^{N_k\times M_l})
\]
is the input space of interference MIMO channels (here, $\prod$ holds for Cartesian product),
\[
\mathcal{S}=\left(\prod_{k\in\Phi_R}\G{d_k}{N_k}\right)\times\left(\prod_{l\in\Phi_T}\G{d_l}{M_l}\right).
\]
is the output space of decoders and precoders (i.e. the set where the possible outputs exist) and
\[
\CV=\{(H,U,V)\in\mathcal{H}\times\mathcal{S}:\text{ (\ref{eq:1}) holds}\}
\]
is the so--called solution variety. $\CV$ is given by certain polynomial equations, linear in each of the $H_{kl},U_{k},V_{l}$ and therefore is an algebraic subvariety of the product space $\mathcal{H}\times\mathcal{S}$.

Note that, given $H\in\mathcal{H}$, the set $\pi_1^{-1}(H)$ is a copy of the set of $U,V$ such that (\ref{eq:1}) holds, that is the solution set of the linear interference alignment problem. On the other hand, given $(U,V)\in\mathcal{S}$, the set $\pi_2^{-1}(U,V)$ is a copy of the set of $H\in\mathcal{H}$ such that (\ref{eq:1}) holds. The feasibility question can then be restated as, {\em is $\pi_1^{-1}(H)\neq\emptyset$ for a generic $H$?}

%In \cite{BreslerCartwrightTseToappear}, \cite{Razaviyayn1} algebraic geometry techniques were used to prove certain properties in the event that the problem is generically feasible. We will generalize some of these results and indeed prove that it is either feasible for every $H$, or infeasible for almost every $H$.
%This leads us to consider the sets $\CV,\mathcal{H},\mathcal{S}$ not only as algebraic sets but also as smooth (compact) manifolds, and to study their tangent spaces.

%--------------------------------------------------------------------------

\section{Characterizing the feasibility of linear IA}

\label{main}

In this section we present a theorem that characterizes the feasibility of linear interference alignment for MIMO channels with constant coefficients for any number of users, antennas and streams per user. This characterization will allow us to provide a polynomial-complexity test of feasibility for this problem which will be detailed in Section \ref{test}.

First, let us fix $d_j,M_j,N_j$ and $\Phi$ satisfying (\ref{eq:p2p_assumption1}) and (\ref{eq:p2p_assumption2}) and define $s\in\Z$ such that
\begin{equation}\label{eq:2}
s=\left( \sum_{k\in\Phi_R}N_kd_k-d_k^2\right)+\left( \sum_{l\in\Phi_T}M_ld_l-d_l^2\right)-\sum_{(k,l)\in\Phi}d_kd_l
\end{equation}
which accounts for the difference between the number of variables and the number of equations in the system of polynomial equations \eqref{eq:1}, as first studied in \cite{Yetis10}. In \cite[Theorem 2]{BreslerCartwrightTseToappear} and \cite[Theorem 1]{Razaviyayn1}, it has been proved that if $s<0$ then, for every choice of $H_{kl}$ out of a zero--measure subset, the system of polynomial equations (\ref{eq:1}) has no solution and, therefore, the IA problem is infeasible. On the other hand, when $s\geq0$, which is the scenario of interest for this paper, the IA problem can be either feasible or infeasible. The situation remains equal in the partially connected case.

\begin{remark}
\label{th1_remark2}
In \cite{Yetis10}, systems were classified as either \textit{proper} or \textit{improper}. A system was deemed \textit{proper} if and only if for every subset of equations in \eqref{eq:1},
the number of variables is at least equal to the number of equations
in that subset. This evaluation may be computationally demanding
with the additional limitation that \textit{properness} is necessary
\cite{BreslerCartwrightTseToappear, Razaviyayn1} but not sufficient for a system to be feasible.  For that reason, in this paper we will follow a simpler convention that classifies a system as $proper$ when $s\ge 0$, which only considers the total set of equations. Our reasoning to define $s$ is based on dimensionality counting arguments whose proof is similar to the ones presented in \cite[Lemma 7]{BreslerCartwrightTseToappear} and which we have omitted herein to avoid repetitions.
\end{remark}
\vspace{0.5cm}
When $s \geq 0$ the following result suggests a practical test to distinguish if, for a choice of $d_j,M_j,N_j,\Phi$, the corresponding linear IA problem is feasible or infeasible.

\begin{theorem}
\label{th:distinct}
Fix $d_j,M_j,N_j$ and $\Phi$ satisfying (\ref{eq:p2p_assumption1}) and (\ref{eq:p2p_assumption2}). Let $s$ be defined by (\ref{eq:2}) and assume that $s\geq0$. Then, the following two cases appear
\begin{enumerate}
\item for every choice of $H_{kl}$ out of a zero--measure subset, the system (\ref{eq:1}) has no solution and, therefore, the IA problem is infeasible; or
\item for every choice of $H_{kl}$ there exists at least one solution to (\ref{eq:1}) and for every choice of $H_{kl}$ out of a zero--measure set the set of solutions of (\ref{eq:1}) is a smooth complex algebraic submanifold; therefore, the IA problem is feasible. In this situation, the following claims are equivalent:

\begin{enumerate}
\item The system (\ref{eq:1}) has solution for every choice of $H_{kl}$.
\item For almost every choice of $H_{kl}$, and for any choice of $U_k,V_l$ satisfying (\ref{eq:1}), the linear mapping
 \par\rlap{\parbox{\columnwidth}
  {
  \begin{align}\label{eq:4}
%\begin{equation}\label{eq:4}
%\begin{split}
\phantom{12345}\theta:&\displaystyle\prod_{k\in\Phi_R}\!\!\mathbb{C}^{N_k\times d_k}\times \displaystyle\prod_{l\in\Phi_T}\!\!\mathbb{C}^{M_l\times d_l}&\ra &\;\;\displaystyle\prod_{(k,l)\in\Phi}\!\!\mathbb{C}^{d_k\times d_l}\notag\\
&(\{\dot{{U}}_k\}_{k\in\Phi_R},\{\dot{{V}}_l\}_{l\in\Phi_T})&\mapsto&\;\;\left\{\dot{{U}}_k^T {H}_{kl}{{V}}_l+\right.\notag\\
&&&\phantom{\left\{\right.}\left.{{U}}_k^T{H}_{kl}\dot{{V}}_l\right\}_{(k,l)\in\Phi}
%\end{matrix}
%\end{equation}
\end{align}
}}
is surjective (i.e. it has maximal rank, equal to $\sum_{(k,l)\in\Phi}d_kd_l$).  Here, we note that some affine representatives $H_{kl}, U_k, V_l$ have been taken.
% and $\left(\dot{\tilde{U}}_1,\ldots,\dot{\tilde{U}}_K,\dot{\tilde{V}}_1,\ldots,\dot{\tilde{V}}_K\right)$ denotes any vector in the tangent space of $\tilde{\mathcal{S}}$ at $(\tilde{U},\tilde{V})$.
\item There exist a $H_{kl}$ and a choice of $U_k,V_l$ satisfying (\ref{eq:1}), such that the linear mapping (\ref{eq:4}) is surjective.
\end{enumerate}
\end{enumerate}
\end{theorem}

\subsection{Geometrical insight behind Theorem \ref{th:distinct}}
\label{sec:geometrical_insight}
A clear understanding of Theorem \ref{th:distinct} comes from considering the solution variety already defined as
\[
\CV=\{(H,U,V):\text{ (\ref{eq:1}) holds}\}.
\]
Consider the projection $\pi_1$ into the first coordinate $H$. Then, an instance $H$ has a solution if and only if $\pi_1^{-1}(H)$ is nonempty. It turns out that both the set $\mathcal{H}$ of inputs $H$ and the set $\CV$ are smooth manifolds. The case $s<0$ will correspond to the dimension of $\CV$ being smaller than that of $\mathcal{H}$, which intuitively implies that the projection of $\CV$ cannot cover the greatest part of $\mathcal{H}$. The case $s\geq0$ will correspond to the dimension of $\CV$ being greater than or equal to that of $\mathcal{H}$. A naive approach should then tell us that the projection of $\CV$ will cover ``at least a good portion'' (i.e. an open subset) of $\mathcal{H}$. Indeed, the algebraic nature of our sets and classical results from differential topology imply that if an open set of $\mathcal{H}$ is reached by the projection, then the whole $\mathcal{H}$ is. This will be the case of item $2)$ of Theorem \ref{th:distinct}. But there is another, counterintuitive thing that can happen: if the {\em whole} set $\CV$ projects into $\mathcal{H}$ in a {\em singular way} (more precisely, if every point of $\mathcal{V}$ is a {\em critical point} of $\pi_1$, namely the tangent space above does not cover the tangent space below), it will still happen that the image of $\CV$ is a zero--measure subset of $\mathcal{H}$, which will produce the case $1)$ of Theorem \ref{th:distinct}. Geometrically, the reader may imagine $\CV$ as a vertical line and $\mathcal{H}$ as a horizontal line: the projection of $\CV$ into $\mathcal{H}$ is just a point, thus a zero--measure set, although both manifolds have the same dimension. This setting looks such a particular situation that it is hard to imagine it happening in real--life examples, but indeed it does happen for many choices of $M_j,N_j,d_j,K$ that are in case $1)$. The good news is that the particular case that all of $\CV$ projects into $\mathcal{H}$ in a singular way, can be easily detected by linear algebra routines involving the mapping (\ref{eq:4}) which is related to the derivative of this projection. This analysis will produce the feasibility test proposed in this paper.

\subsection{Extensions and discussion of related results}
\label{sec:extensions_and_related_results}
Let us point out that the model we have used for our derivations, i.e. diagram \eqref{eq:diag}, is similar to that used in \cite[Section 2]{BreslerCartwrightTseToappear}. The only difference is that in our case we let channels live in the projective space of matrices which is a compact space instead of the non-compact affine space used in \cite{BreslerCartwrightTseToappear}. The arguments that lead to the proof that a system is infeasible when $s<0$ are based on the dimensionality of the solution variety \cite[Lemmas 7, 8]{BreslerCartwrightTseToappear}. %We have omitted such details to avoid unnecesary repetitions.

The fact that either almost every $H_{kl}$ admits a solution or almost every $H_{kl}$ does not admit a solution, was essentially proved in \cite{BreslerCartwrightTseToappear} and \cite{Razaviyayn1}. The constructions of the Zariski cotangent space in \cite{BreslerCartwrightTseToappear}, the Jacobian computation in \cite{Razaviyayn1} and the matrix in \cite{Ruan2013} are strongly related to that of the mapping \eqref{eq:4}. One difference is that the derivation of \eqref{eq:4} does not require any particularization or partitioning of the factors appearing in the alignment equations \eqref{eq:1}, as done in \cite{BreslerCartwrightTseToappear} and \cite{Razaviyayn1}, respectively. Instead, it has been derived (independently of the chosen representatives) as a mapping between tangent spaces, which endows our approach with the simple geometrical interpretation provided in Section \ref{sec:geometrical_insight}.

Furthermore, despite the obvious connections with \cite{BreslerCartwrightTseToappear} and \cite{Razaviyayn1}, the tools and mathematical framework used in this paper allowed us to prove that, when the system is feasible and $s=0$, then the number of IA solutions is finite and constant for almost all channel realizations. This is formally stated in the following lemma.
\begin{lemma}
\label{lem:numbersolutions}
For almost every $H$, the solution set in case 2) of Theorem \ref{th:distinct} is a smooth complex algebraic submanifold of dimension $s$. If $s=0$, then there is a constant $C\ge 1$ such that for every choice of $H_{kl}$ out of a proper algebraic subvariety (thus, for every choice out of a zero measure set) the system \eqref{eq:1} has exactly $C$ alignment solutions.
\end{lemma}

\begin{IEEEproof}
See Section \ref{sec:proofs}.
\end{IEEEproof}

A practical consequence of Lemma \ref{lem:numbersolutions} is that affine alignment solutions (when finite) are grouped in $C$ orbits of equivalent solutions spanning the same subspace. This fact is automatically captured by the way we have modeled the output space $\mathcal{S}$ that considers precoders and decoders as Grassmannians and therefore enables us to see those orbits as $C$ isolated solutions.
\begin{remark}
As pointed out in Section
\ref{sec:problem_statement}, if some $k_0$ satisfies $k_0\not\in\Phi_R$ or some $l_0$ satisfies $l_0\not\in\Phi_T$, then any solution $(\{U_k\}_{k\in\Phi_R},\{V_l\}_{l\in\Phi_T})$ can be complemented with any choice of $U_{k_0}$ and $V_{l_0}$ and still be a solution of (\ref{eq:1}), just because the variables $U_{k_0}$ and $V_{l_0}$ do not appear in (\ref{eq:1}). When we say that the number of solutions is a finite number $C$, we are not counting these infinitely many possible choices for $U_{k_0}$ and $V_{l_0}$. We trust that this convention is clear and natural enough to avoid confusion.
\end{remark}

%-------------------------------------------

\section{Proposed feasibility test}
\label{test}
\subsection{A floating-point arithmetic test of feasibility}
\label{sec:floatingpoint_test}
We now construct a test for checking whether a given choice of $d_j,M_j,N_j,\Phi$ defines a feasible alignment problem or not. To develop this test, we first have to choose a point $H_{kl}, U_k, V_l$ such that (\ref{eq:1}) holds. An arbitrary set of channels, decoders and precoders satisfying the IA equations (\ref{eq:1}) can be obtained very easily by solving what we call the {\it inverse} IA problem; that is, given a set of arbitrary (e.g. random) decoders and precoders, $U_k,V_l$, find a set of MIMO channels such that (\ref{eq:1}) holds. This is totally different from (and much easier to solve than) the original IA problem, which is given channel matrices $H_{kl}$, find elements $U_k,V_l$ that solve (\ref{eq:1}). Since the polynomial equations (\ref{eq:1}) are linear in $H_{kl}$ the {\it inverse} IA problem is completely solved by the following Lemma.
\begin{lemma}\label{lem:particularH}
Fix any choice of $d_j,M_j,N_j,\Phi$ satisfying (\ref{eq:p2p_assumption1}) and (\ref{eq:p2p_assumption2}), and let $(U,V)\in\mathcal{S}$ be any element. Then, the set
\[
\pi_2^{-1}(U,V)=\{H\in\mathcal{H}:(H,U,V) \text{ solve (\ref{eq:1}) }\}\subseteq\mathcal{H}
\]
is a nonempty product of projective vector subspaces and a smooth submanifold of $\mathcal{H}$ of complex dimension equal to
\[
\left(\sum_{(k,l)\in\Phi}N_kM_l-d_kd_l\right)-\sharp(\Phi).
\]
In particular, this quantity is greater than or equal to $0$.
%(Indeed, a stronger inequality
%\[
%N_kM_l>d_kd_l,\quad (k,l)\in\Phi,
%\]
%can be proved following the lines of our Section \ref{sec:IAinverse}).
\end{lemma}
\begin{IEEEproof}
See Appendix \ref{Appendix2}.
\end{IEEEproof}
 Lemma \ref{lem:particularH} shows that we may fix our $U$ and $V$ to be the ones of our choice and there always exists $H$ forming a valid element $(H,U,V)\in\CV$. If, for that choice, the linear mapping defined in (\ref{eq:4}) is surjective, then the alignment problem is generically feasible by item (2.c) of Theorem \ref{th:distinct}. If for generic $H$ that mapping is not surjective the alignment problem is not generically feasible, namely it can be solved just for a zero--measure set of $H_{kl}$. The proposed feasibility test then has to perform two tasks:
\begin{enumerate}
\item Find an arbitrary $H_{kl},U_k,V_l$ such that (\ref{eq:1}) holds. We will detail later a simple choice for these elements.
\item To check whether the matrix $\Psi$ (in any basis) defining the linear mapping (\ref{eq:4}) satisfies $\det(\Psi\Psi^*)\neq 0$ (which is equivalent to mapping $\theta$ defined in (\ref{eq:4}) being surjective) or not.
%Compute the rank of (\ref{eq:4}). This is a linear algebra routine and is hence easily doable. The size of the matrices involved is polynomial in $K,M_j,M_j,d_j$.
\end{enumerate}

Now, we detail the two stages of the proposed IA feasibility test.

\subsubsection{Finding an arbitrary IA solution}\label{sec:IAinverse}
The first stage requires finding arbitrary $U_k,V_l$ and their corresponding MIMO channels $H_{kl}$ such that \eqref{eq:1} holds.  Lemma \ref{lem:particularH} allows us to choose any $U_k$ and $V_l$ of our choice. Thus, we will consider precoders and decoders given by
\begin{equation}\label{eq:form1}
V_l=\binom{I_{d_l}}{0_{(M_l-d_l)\times d_l}}, \hspace{0.5cm}  U_k=\binom{I_{d_k}}{0_{(N_k-d_k)\times d_k}},
\end{equation}
and MIMO channels with the following structure
\begin{equation}\label{eq:form2}
 H_{kl}=\begin{pmatrix}0_{d_k\times d_l}&A_{kl}\\B_{kl}&C_{kl}\end{pmatrix},
\end{equation}
which trivially satisfy $U_k^T H_{kl}V_l=0$ and therefore belong to the solution variety.
We claim that essentially all the useful information about $\CV$ can be obtained from the subset of $\CV$ consisting on triples $(H_{kl},U_k,V_l)$ of the form (\ref{eq:form1}) and (\ref{eq:form2}). The reason is that given any other element $(H'_{kl},U'_k,V'_l)\in\CV$, one can easily find sets of orthogonal matrices $P_k$ and $Q_l$ satisfying
\[
U_k=P_kU'_k,\hspace{0.5cm} V_l=Q_lV'_l,
\]
and
\[
0={U'}^T_k H'_{kl}V'_l=U_k^T\left(P_k^*\right)^TH'_{kl}Q_l^*V_l,
\]
where the superscript $*$ denotes Hermitian. That is, the transformed channels $H_{kl}=\left(P_k^*\right)^TH'_{kl}Q_l^*$ have the form (\ref{eq:form2}), and the transformed precoders $V_l$ and decoders $U_k$ have the form (\ref{eq:form1}).

\subsubsection{Checking the rank of the linear mapping $\theta$}
For a particular element of the solution variety chosen as in (\ref{eq:form1}) and (\ref{eq:form2}), the linear mapping $\theta$ reduces to
\begin{equation}
\label{eq:linear_mapping}
\theta: \hspace{0.2cm} (\{\dot{{U}}_k\}_{k\in\Phi_R},\{\dot{{V}}_l\}_{l\in\Phi_T}) \hspace{0.2cm} \mapsto  \hspace{0.2cm}\left\{\dot{{U}}_k^T {B}_{kl}+{A}_{kl}\dot{{V}}_l\right\}_{(k,l)\in\Phi},
\end{equation}
where $\dot{{U}}_k$, $\dot{{V}}_l$ have dimensions $(N_k-d_k)\times d_k$ and $(M_l-d_l)\times d_l$, respectively. The mapping $\theta$ can also be written in matrix form as
\begin{equation}
\Psi w,
\end{equation}
where $w$ is a column vector of dimension $\sum_{k\in \Phi_R}(N_k-d_k)d_k+\sum_{l\in \Phi_T}(M_l-d_l)d_l$, built by stacking all columns of $\{\dot{{U}}_k^T\}_{k\in\Phi_R}$ and $\{\dot{{V}}_l^T\}_{l\in\Phi_T}^T$,
 and $\Psi$ is a block matrix with $\sharp(\Phi)$ row partitions (as many blocks as interfering links) and $2K$ column partitions (as many blocks as precoding and decoding matrices). Checking the feasibility of IA then reduces to check whether matrix $\Psi$ is full rank or not. Vectorization of the mapping \eqref{eq:linear_mapping} reveals that $\Psi$ is composed of two main kinds of blocks, $\Psi^{(A)}_{kl}$ and $\Psi^{(B)}_{kl}$, i.e.
\begin{equation}
\begin{split}
\label{eq:Psi_A_Psi_B}
\myvec(\dot{{U}}_k^T {B}_{kl}+{A}_{kl}\dot{{V}}_l)&=%\cr
\overbrace{(A_{kl}\otimes I_{d_k})K_{(N_k-d_k),d_k}}^{\displaystyle \Psi_{kl}^{(A)}}\myvec(\dot{U}_k)\\&+\underbrace{(I_{d_l}\otimes B_{kl}^T)}_{\displaystyle \Psi_{kl}^{(B)}}\myvec(\dot{V}_l),
\end{split}
\end{equation}
where $\otimes$ denotes Kronecker product and $K_{m,n}$ is the $mn \times mn$ commutation matrix which is defined as the matrix that transforms the vectorized form of an $m\times n$ matrix into the vectorized form of its transpose. Block $\Psi^{(B)}_{kl}$ has dimensions $d_ld_k\times d_l(M_l-d_l)$, whereas block $\Psi^{(A)}_{kl}$ is $d_ld_k\times d_k(N_k-d_k)$. For a given tuple $(k,l)$, $\Psi^{(B)}_{kl}$ and $\Psi^{(A)}_{kl}$ are placed in the row partition that corresponds to the interfering link indicated by the tuple $(k,l)$. $\Psi^{(B)}_{kl}$ is placed in the $l+K$-th column partition, whereas $\Psi^{(A)}_{kl}$ occupies the $k$-th column partition. The rest of blocks are occupied by null matrices. The dimensions of $\Psi$ are therefore
\begin{equation*}
\sum_{(k,l)\in\Phi}d_kd_l\times \sum_{k\in \Phi_R}(N_k-d_k)d_k+\sum_{l\in \Phi_T}(M_l-d_l)d_l,
\end{equation*}
whereas its structure is exactly the same as the incidence matrix of the network connectivity graph. Remarkably, in the particular case of $s=0$, $\Psi$ is a square matrix of size $\sum_{(k,l)\in\Phi}d_kd_l$.

Taking the $3$-user interference channel as an example, $\Psi$ is given as in \eqref{eq:linear_mapping_3user}, where the blocks $\Psi^{(B)}_{kl}$ and $\Psi^{(A)}_{kl}$ are given by \eqref{eq:Psi_A_Psi_B}.
\begin{figure*}[t]
\begin{equation}
\label{eq:linear_mapping_3user}
\begin{array}{ccc}
&& \text{Column partition}\\
\parbox[c]{2cm}{\centering Interfering\\link}&\parbox[c]{2cm}{\centering Row\\partition}&{\setlength\arraycolsep{1.2em}\begin{array}{cccccc}
1&2&3&4&5&6
\end{array}}\\
\begin{array}{c}
(1,2)\\ (1,3) \\(2,1)\\(2,3)\\(3,1)\\(3,2)
\end{array}&\begin{array}{c}
1\\ 2 \\3\\4\\5\\6
\end{array} & \left[ \begin{array}{cccccc}
\Psi^{(A)}_{12}&0&0&0&\Psi^{(B)}_{12}&0\\
\Psi^{(A)}_{13}&0&0&0&0&\Psi^{(B)}_{13}\\
0&\Psi^{(A)}_{21}&0&\Psi^{(B)}_{21}&0&0\\
0&\Psi^{(A)}_{23}&0&0&0&\Psi^{(B)}_{23}\\
0&0&\Psi^{(A)}_{31}&\Psi^{(B)}_{31}&0&0\\
0&0&\Psi^{(A)}_{32}&0&\Psi^{(B)}_{32}&0
\end{array}\right]
\end{array}
\end{equation}
\hrulefill
\end{figure*}

 Once $\Psi$ has been built, the last step is to check whether the mapping is surjective and, consequently, the interference alignment problem is feasible. This amounts to check if the rank of $\Psi$ is maximum, that is, equal to $\sum_{(k,l)\in\Phi}d_kd_l$. A simple method consists of generating a random element $b \in \C^{\sum_{{(k,l)\in\Phi}} d_kd_l}$, computing the least squares solution of $\Psi w=b$ and checking if $\|\Psi w-b\|$ is below a given threshold $\mu$. With a high probability in the choice of $b$ this test will determine if $\theta$ is a surjective mapping.

At this point, two questions regarding the practical implementation of this method may arise. The first one is related to the scalability of the proposed method. It is obvious that both the computational and storage requirements grow with the number of antennas, streams and users in the system. However, matrix $\Psi$ presents two characteristics which limit, to some extent, these requirements.

\begin{itemize}
\item First, $\Psi$ is a very sparse matrix with only $\sum_{(l,k)\in\Phi}(N_k-d_k)d_ld_k + \sum_{(k,l)\in\Phi}(M_l-d_l)d_kd_l$ non-zero entries, thus limiting both the computational and the storage requirements. Sparsity can be exploited by computing the least squares solution of $\Psi w=b$ from the sparse QR factorization of $\Psi$, for which efficient algorithms exist \cite{Davis2011}.
\item Recall also that the matrix-vector product $\Psi w$ is completely characterized by the entries of submatrices $A_{kl}$ and $B_{kl}$ in \eqref{eq:form2}. Black box iterative algorithms \cite{Paige1982} are able to solve the least squares problem by solely performing matrix-vector products, i.e. computing the linear transformation defined by the matrix $\Psi$. The main consequence of this is that $\Psi$ does not even need to be explicitly constructed thus reducing even further the storage requirements.
\end{itemize}
These considerations allowed us to evaluate the feasibility of systems whose resulting $\Psi$ is of dimensions up to $40000\times 40000$. As a rule of thumb, we could say that symmetric systems with a product $Kd$ up to $200$ are computable. As an example, we were able to check that the system $(86\times 139,25)^8$ is feasible. This operating range allowed us to extensively verify the feasibility of a wide variety of scenarios and even establish a new conjecture regarding the DoF of symmetric interference channels which is described in detail in Section \ref{sec:simulations}.

The second question refers to the reliability of the numerical results. Floating-point algorithms are always prone to round-off errors, hence, determining something as simple as the rank of a matrix may not be that easy, especially for very large systems. The choice of the threshold $\mu$ determines in the end to which extent our results are reliable. To eliminate this ambiguity, in Section \ref{sec:exact_test} we present a Turing machine, exact arithmetic, version of the proposed test and prove that checking the IA feasibility belongs to the complexity class of bounded-error probabilistic polynomial time (BPP) problems. From a practical point of view, however, the floating point version of the test described in this section was found to provide always robust and consistent results when the entries in $A_{kl}, B_{kl}$ and $w$ were drawn from a complex normal distribution with zero mean and unit variance, and the decision threshold was set to $\mu=10^{-3}$.

\subsection{Exact arithmetic test and complexity analysis}
\label{sec:exact_test}
The test explained after our Theorem \ref{th:distinct} has been programmed in floating point arithmetic, and it is thus sensitive to floating point errors. Although it is robust enough for many examples, a Turing machine version of this test (that is, a test working in exact arithmetic) is in order. Consider the following algorithm.
\begin{enumerate}
\item For $k\in\Phi_R$ and $l\in\Phi_T$, consider $H_{kl}$ as in (\ref{eq:form2}). Let $C_{kl}=0$ for all $k,l$ and let the entries of $A_{kl}$ and $B_{kl}$ be chosen (i.i.d uniformly) as $a+\sqrt{-1}b$ where $0\leq a,b< h$, $a,b\in\Z$, and
\[
h=8\sum_{(k,l)\in\Phi}(N_k-d_k)d_l+(M_l-d_l)d_k.
\]
Note thus that the entries of $A_{kl},B_{kl}$ are complex numbers whose real and imaginary parts are integers of bounded size, chosen at random.
\item  Check, using exact linear algebra procedures (such as the ones available in libraries IML \cite{Chen2005} or LinBox \cite{Dumas2002}), if the mapping \eqref{eq:linear_mapping} is surjective. Then,
\begin{itemize}
\item if the mapping is surjective, answer {\em feasible},
\item otherwise, answer {\em infeasible}.
\end{itemize}
\end{enumerate}

The following is our second main result.

\begin{theorem}\label{th:Turing}
The algorithm above is a Bounded Error Probability procedure (thus, describes a {\bf BPP} Turing machine) whose running time is polynomial in the input parameters $d_j,M_j,N_j,\sharp(\Phi)$:
\begin{itemize}
\item if the given parameters define a unfeasible alignment problem, answers {\em unfeasible}.
\item if the given parameters define a feasible alignment problem, with high probability the algorithm answers {\em feasible}, but there is a probability (in the choice of the coefficients of $A_{kl},B_{kl}$) of at most $1/4$ that the algorithm answers {\em unfeasible}.
\end{itemize}
\end{theorem}

\begin{IEEEproof}
See Section \ref{sec:proofs}. Here is an outline of the idea of the proof: if the scenario is feasible, then for every choice of $H_{kl}$ out of some zero measure set $\mathcal{Z}$, the mapping \eqref{eq:linear_mapping} is surjective. Of course, it could happen that every choice of $H_{kl}$ with integer, ``small'' entries is in $\mathcal{Z}$. But, for that to happen, $\mathcal{Z}$ must have a complicated topology (think for example in a line that touches all points in the $xy$--plane with integer coordinates bounded by some $h>0$: the line must have quite a complicated shape). But, the shape of $\mathcal{Z}$ is actually very simple because it is given by a set of multilinear equations of small degree. Thus, $\mathcal{Z}$ cannot contain too many integer points, and as a consequence for ``most'' integer points, the mapping in \eqref{eq:linear_mapping} must be surjective.
\end{IEEEproof}

Note that this kind of algorithm (with a bounded error probability just in one direction) is very common in mathematics (the most famous example is Miller--Rabin test for primality \cite{Miller1976,Ra80}). The use is very simple: if on a given input the algorithm answers {\em feasible} then the alignment is feasible. If the test is run $k$ times and its answer is {\em unfeasible} for all $k$ tries, then we can conclude that the alignment is unfeasible unless an extremely unlikely event (probability at most $1/4^k$) happened. The upper bound $1/4$ on the error probability in the one--try test can be changed to any $\epsilon<1$ by choosing a different value of $h$, but according to the previous discussion, the specific value is irrelevant.

Technically, Theorem \ref{th:Turing} asserts that the problem of deciding if a given choice of $d_j,M_j,N_j,\Phi$ is generically infeasible is in the complexity class BPP (bounded-error probabilistic polynomial time).

\begin{remark}
\label{th2_remark1}
Some complexity analysis have recently appeared in the literature claiming that to check the feasibility of IA problems is strongly NP-hard \cite{Razaviyayn2},\cite{Luo11}. However, there is a crucial difference between the problem considered in \cite{Razaviyayn2},\cite{Luo11} and that considered in this paper. The problem in \cite{Razaviyayn2} can be restated informally as follows:
\begin{problem}
\label{problem1}
Given $d_j$, $M_j$ and $N_j$, decide whether there exists a linear alignment solution for a {\it given} set of interference MIMO channels $H_{kl}$.
\end{problem}
However, we are considering in this paper a different feasibility problem:
\begin{problem}
\label{problem2}
Given $d_j$, $M_j$ and $N_j$ (and a connectivity graph or matrix $\Phi$), decide whether there exists a linear alignment solution for {\it generic} interference MIMO channels $H_{kl}$.
\end{problem}
While Problem \ref{problem1} is NP-hard, we have just shown that Problem \ref{problem2} can be solved in polynomial time. The complexity of Problem \ref{problem1} is due to the fact the authors in \cite{Razaviyayn2},\cite{Luo11} consider a given realization of $H_{kl}$. In fact, to check whether this channel realization admits a solution, can indeed be NP-hard. However, by restricting the problem to generic MIMO channels, e.g., channels with independent entries drawn from continuous distributions, the IA feasibility problem becomes much easier. Note also that even if checking the feasibility of IA can be done with polynomial complexity, finding the actual decoders and precoders that align the interference subspaces can still be NP-hard when $K$ is large, as proved in \cite{Razaviyayn2}.

\end{remark}

\begin{remark}
\label{th2_remark2}
The IA feasibility problem considered in this paper, that is, determining if a given stream distribution $\left( d_1,\ldots,d_K \right)$ can be generically achieved with linear IA is tightly related to that of finding the maximum total DoF (or the tuple achieving the maximum sum DoF, $\sum_{k=1}^K d_k$). Although we have shown that the former belongs to the BPP class, the complexity of the latter remains uncharacterized. Based on the proposed test, we have recently presented an algorithm \cite{Gonzalez2013} to compute the maximum DoF in arbitrary networks. Its working principle is performing an ordered search inside the region of potential feasible tuples (those which satisfy existing necessary conditions) until a feasible tuple is found. Unfortunately, for an arbitrary system, the accurate determination of this region may be a computationally demanding task. A problem of similar complexity is that of checking the necessary feasibility conditions in \cite[Theorem 2]{BreslerCartwrightTseToappear} or \cite[Theorem 1]{Razaviyayn1}, which involve an exponential number of constraints.
\end{remark}
%--------------------------------------------------------------------

\section{Proof of Main Results}
\label{sec:proofs}

%We must also point out that some --but not all-- of our results can be proved without the use of so much machinery from differential topology (using instead machinery from algebraic geometry as Bertini's Theorem, see for example \cite{Sh94a}). However, the present proof produces stronger results (for example, the fiber--bundle structure of the solution variety) and allows us to very easily derive the test of Theorem \ref{th:distinct}.

%It is important for our analysis that the input and output spaces are defined over the complex numbers, not over the reals. Indeed, a key property in proving our main results is that the critical points and values of
%$\pi_1$ are algebraic sets. In the complex case this means they have (real) codimension $2$ and hence do not disconnect their ambient spaces.
%In the real case, these sets have real codimension $1$ and they may thus disconnect their ambient spaces. More specifically, it is Corollary
%\ref{cor:connected} below that is false in the real case. As a consequence, one cannot apply Ehresman's Theorem and a more delicate analysis is required.

%---------------------------------------------------------------------

In what follows we provide a rigorous proof of our results.
% by first presenting some Lemmas on the dimension of the algebraic sets involved in the problem. Second, we consider their topological properties and present some results related to the critical points and critical values of projection $\pi_1$. Finally, we present the proofs of theorems \ref{th:distinct} and \ref{th:Turing}.
Most preliminary details of the proof are relegated to appendices.

\subsection{Dimensions of the algebraic manifolds involved in the problem}

In this subsection we recall the dimensions of the algebraic sets involved in the problem. Similar results have appeared in \cite{Yetis10}, \cite{BreslerCartwrightTseToappear} and \cite{Razaviyayn1}; therefore and to keep the paper concise, their proofs are omitted. For the interested reader the proofs can be deduced following the mentioned references \cite{Yetis10}, \cite{BreslerCartwrightTseToappear}, \cite{Razaviyayn1} with a basic knowledge of algebraic geometry tools such as those described in \cite{Sh94a} and \cite{Whitney1972}.

\begin{lemma}\label{lem:HS}
Both $\mathcal{H}$ and $\mathcal{S}$ are complex manifolds, and
\[
\dim_\C\mathcal{H}=\sum_{(k,l)\in\Phi}(N_kM_l-1),\]
\[
\dim_\C\mathcal{S}=\sum_{k\in\Phi_R}d_k(N_k-d_k)+\sum_{l\in\Phi_T}d_l(M_l-d_l).
\]
\end{lemma}

%Finally, let
%\[
%\hat{\CV}=\{(i,o)\in\hat{\mathcal{H}}\times\hat{\mathcal{S}}:\text{ (\ref{eq:1}) holds}\}.
%\]
%Proving that $\CV$ is also a manifold is not so straight--forward.

\begin{lemma}\label{lem:V}
The set $\CV$ is a complex smooth submanifold of $\mathcal{H}\times\mathcal{S}$ and its complex dimension is
\begin{align*}
\dim_\C\CV=\left(\sum_{(k,l)\in\Phi}N_kM_l-d_kd_l\right)+\left(\sum_{k\in\Phi_R}N_kd_k-d_k^2\right)\\+\left(\sum_{l\in\Phi_T}M_ld_l-d_l^2\right)-\sharp(\Phi).
\end{align*}
\end{lemma}
%---------------------------------------
\subsection{The critical points and values of $\pi_1$}\label{sec:alggeom}
We now study the sets of critical points and values of $\pi_1$.
\begin{lemma}\label{lem:aux}
Let $(H,U,V)\in\CV$ be fixed and let $\theta$ be the mapping defined in (\ref{eq:4}). Then, $\theta$ is surjective or not, independently of the chosen representatives of $(H,U,V)$.
\end{lemma}
%-------------------------------------------------------------
\begin{IEEEproof}
See Appendix \ref{Appendix6}.
\end{IEEEproof}
%-----------------------------------------------------------------------
\begin{proposition}\label{prop:charsigma}
Let $(H,U,V)\in\CV$. Then, $(H,U,V)$ is a regular point of $\pi_1$ if and only if the mapping $\theta$ defined in (\ref{eq:4}) is surjective.
\end{proposition}
\begin{IEEEproof}
See Appendix \ref{Appendix7}.
\end{IEEEproof}

%------------------------------------------------------------------------

\begin{proposition}\label{prop:sigma}
The set $\Sigma'\subseteq\CV$ of critical points of $\pi_1$ is an algebraic subvariety of $\CV$. The set $\Sigma\subseteq \mathcal{H}$ of critical values of $\pi_1$ is a proper (i.e. different from the total) algebraic subvariety of $\mathcal{H}$.
\end{proposition}
%-------------------------------------------------
\begin{IEEEproof}
See Appendix \ref{Appendix8}.
\end{IEEEproof}
%--------------------------------------------------------------------------------------

\begin{corollary}\label{cor:connected}
$\mathcal{H}\setminus\Sigma$ is a connected set.
\end{corollary}
\begin{IEEEproof}
From Proposition \ref{prop:sigma}, the set $\Sigma$ is a complex proper algebraic subvariety, therefore it has real codimension $2$ and removing it does not disconnect the space $\mathcal{H}$.
\end{IEEEproof}
%----------------------------------------------------------------------------------
\begin{corollary}\label{cor:ehr}
Assume that $\Sigma'$ is a proper algebraic subvariety of $\CV$ (equivalently, $\pi_1:\CV\ra\mathcal{H}$ has at least one regular point). Then, we are in the case $2)$ of our Theorem \ref{th:distinct}, that is for every $H\in\mathcal{H}$ the set $\pi_1^{-1}(H)$ is nonempty, and for $H\not\in\Sigma$ it is a smooth complex manifold of dimension $s$. Indeed, the restriction $V\setminus\pi_1^{-1}(\Sigma)\overset{\pi_1}{\ra}\mathcal{H}\setminus\Sigma$ is a fiber bundle.
\end{corollary}
%---------------------------------------------------------------------------
\begin{IEEEproof}
From Corollary \ref{cor:connected},
%$\Sigma$ is an algebraic subvariety of $\mathcal{H}$ and
the set $\mathcal{H}\setminus\Sigma$ of non-critical values of $\pi_1$ is a connected set. Moreover, we have:
\begin{itemize}
\item $\CV\setminus\pi_1^{-1}(\Sigma)$ is not empty by assumption,
\item $\pi_1\mid_{\CV\setminus\pi_1^{-1}(\Sigma)}$ is a submersion (because we have removed the set of critical points), and
\item it is proper: let $A\subseteq\mathcal{H}\setminus\Sigma\subseteq\mathcal{H}$ be a compact set. Then, $A$ is closed as a subset of $\mathcal{H}$ and from the continuity of $\pi_1$, so is $A'=\pi_1^{-1}(A)\subseteq \CV$. Now, $A'$ is a closed subset of the compact set $\CV$ and hence $A'$ is compact.
\end{itemize}
Ehresmann's Theorem then implies that $\pi\mid_{\CV\setminus\pi_1^{-1}(\Sigma)}$ is a fiber bundle, and in particular it is surjective. This proves that $\pi_1^{-1}(H)\neq\emptyset$ for every $H\in\mathcal{H}\setminus\Sigma$, and the Preimage Theorem implies that $\pi_1^{-1}(H)$ is a smooth submanifold of complex codimension equal to $\dim_\C\mathcal{H}$, thus of complex dimension equal to $\dim_\C\CV-\dim_\C\mathcal{H}=s$. Now, let $H\in\Sigma$ and let $H_i,i\geq1$ be a sequence of elements in $\mathcal{H}\setminus\Sigma$ such that $\lim_{i\mapsto\infty}H_i=H$. Let $(H_\infty,U_\infty,V_\infty)$ be an accumulation point of $(H_i,U_i,V_i)\in\CV$, which exists because $\CV$ is compact. Then, by continuity of $\pi_1$ we have that $\pi_1(H_\infty,U_\infty,V_\infty)=H$, that is $H_\infty=H$ and $(H_,U_\infty,V_\infty)\in\CV$. Thus, $\pi_1^{-1}(H)\neq\emptyset$ and we conclude that for every choice of $H_{kl}$ there exists at least one solution to (\ref{eq:1}) as claimed.
\end{IEEEproof}

\subsection{Proof of Theorem \ref{th:distinct}}

Recall from Lemma \ref{lem:HS} that the complex dimension of $\mathcal{H}$ is
\[
\dim_\C(\mathcal{H})=\sum_{(k,l)\in\Phi}(N_kM_l-1)=\sum_{(k,l)\in\Phi}N_kM_l-\sharp(\Phi).
\]
From this and from Lemma \ref{lem:V}, defining $s$ as in (\ref{eq:2}) we have
\[
s=\dim_\C\CV-\dim_\C\mathcal{H}.
\]
Assume that $\dim_\C(\mathcal{H})\leq \dim_\C(\CV)$ (equivalently, $s\geq0$). There are two cases:
\begin{enumerate}
\item if $\Sigma'=\CV$ then every point of $\CV$ is a critical point of $\pi_1$ and hence every element of $\pi_1(\CV)$ is a critical value of $\pi_1$. On the other hand, from Proposition \ref{prop:sigma}, $\Sigma$ is a proper algebraic subset of $\mathcal{H}$, thus a zero measure set of $\mathcal{H}$. This means that $\pi^{-1}(H)=\emptyset$ for every $H$ out of the zero--measure set $\Sigma$, thus we are in case 1) of Theorem \ref{th:distinct}.
\item otherwise, $\Sigma'$ is a proper subset of $\CV$, and from Corollary \ref{cor:ehr} we are in case 2) of Theorem \ref{th:distinct}.

We now prove each of the following implications:
\begin{enumerate}

\item[a)$\Rightarrow$b):] assume that $\pi_1^{-1}(H)\neq\emptyset$ for every $H\in\mathcal{H}$ . From Sard's theorem, for almost every $H\in\mathcal{H}$, $\pi_1$ is a submersion at every point in $\pi_1^{-1}(H)$ and from Proposition \ref{prop:charsigma} the mapping (\ref{eq:4}) defines a surjective linear mapping.
\item[b)$\Rightarrow$c):] trivial.
\item[c)$\Rightarrow$a):] from Proposition \ref{prop:charsigma}, $\pi_1$ has a regular point, and from Corollary \ref{cor:ehr}, a) holds.
\end{enumerate}

\end{enumerate}

This finishes the proof.\hfill\IEEEQEDclosed

Finally, the proof of Lemma \ref{lem:numbersolutions} stating when a feasible IA problem has a finite number of solutions is as follows: assume that $s=0$, or equivalently $\dim_\C(\mathcal{H})=\dim_\C(\CV)$, and that we are still in case $2(b)$ of Theorem \ref{th:distinct}. Then, from Corollary \ref{cor:covering} (see Appendix \ref{app:difftop}) all the elements in $\mathcal{H}\setminus\Sigma$ have the same (finite) number, say $C$, of preimages by $\pi_1$. This proves the assertion of Lemma \ref{lem:numbersolutions}.

\begin{remark}
\label{th2_remarknew}
It is important for our analysis that the input and output spaces are defined over the complex numbers, not over the reals. Indeed, a key property in proving our main results is that the critical points and values of
$\pi_1$ are algebraic sets. In the complex case this means they have (real) codimension $2$ and hence do not disconnect their ambient spaces.
In the real case, these sets may have real codimension $1$ and they may thus disconnect their ambient spaces. More specifically, Corollary
\ref{cor:connected} may fail to hold in the real case. As a consequence, one cannot apply Ehresman's Theorem and a more delicate analysis would be required in this case.
\end{remark}

\subsection{Proof of Theorem \ref{th:Turing}}
Assume that parameters $d_j,M_j,N_j,\Phi$ are chosen such that the associated MIMO scenario is feasible. First, let us remind from Section \ref{sec:IAinverse} that we may choose $U_k$ and $V_l$ as those in \eqref{eq:form1}, and the MIMO channels as in \eqref{eq:form2} which, for convenience, we show again:
\[
H_{kl}=\begin{pmatrix}0_{d_k\times d_l}&A_{kl}\\B_{kl}&C_{kl}\end{pmatrix},\qquad (k,l)\in\Phi.
\]
Now, let $h\geq1$ be an integer number and let those matrices have coefficients of the form
\begin{equation}
\frac{a}{h}+\sqrt{-1}\frac{b}{h},
\end{equation}
with denominator $h$ and numerators $a,b$ in $[0,h)\cap \mathbb{Z}$. As the system is generically feasible, for most choices of these matrices $A_{kl},B_{kl},C_{kl}$, we will have $(H,U,V)\not\in\Sigma$, that is the linear mapping in \eqref{eq:linear_mapping} will be surjective. Moreover, the mapping in \eqref{eq:linear_mapping} is independent of the entries $C_{kl}$, so we can simply say that for most choices of $A_{kl},B_{kl}$ the mapping will be surjective. The merit of Theorem \ref{th:Turing} is to quantify this ``for most'', which we do following the arguments in \cite[Sec. 17.4]{BlCuShSm98}, which in turn are inspired by a celebrated result by Milnor bounding the number of connected components of semi--algebraic sets. We start by studying the set
\begin{align*}
\mathcal{Z}=\{(A_{kl},B_{kl})\in [0,1)^{2\sum_{(k,l)\in\Phi}(N_k-d_k)d_l+(M_l-d_l)d_k}:\\ \text{ the linear mapping in \eqref{eq:linear_mapping} is not surjective}\}.
\end{align*}
Note that we consider $\mathcal{Z}$ as a subset of $[0,1)^{2\sum_{(k,l)\in\Phi}(N_k-d_k)d_l+(M_l-d_l)d_k}$, that is a real set, by considering the real and complex parts of each entry of each $A_{kl}$ and $B_{kl}$ as a real number in $[0,1)$.
\begin{lemma}\label{lem:Milnor}
Let $\kappa(\mathcal{Z})$ be the maximum number of connected components (intervals) of $\mathcal{Z}\cap L$ where $L$ is some line parallel to some axis. Then,
\[
\kappa(\mathcal{Z})\leq 1.
\]
\end{lemma}
\begin{IEEEproof}
Let $L$ be a line parallel to some axis. That is, $L$ is the set of all $(A_{kl},B_{kl}),\;(k,l)\in\Phi$, such that all entries of $A_{kl}$ and $B_{kl}$ are fixed save for one of them (the real or the complex part of some entry, call it $\lambda$, of some $A_{kl}$ or some $B_{kl}$). The set $\mathcal{Z}\cap L$ is defined by $\rank(\theta)<\sum_{(k,l)\in\Phi}d_kd_l$, equivalently it is given by
\[
p(\lambda)=\sum_{J}|\det(M)|^2=0,
\]
where $J$ runs over all the possible minors of maximal size contained in the matrix of mapping \eqref{eq:linear_mapping} and $\det(M)$ are those minors. This is thus one real, non--negative equation of degree at most $2$ in $\lambda$. There are several possibilities:
\begin{itemize}
\item Case $p(\lambda)=0$: the set $\mathcal{Z}\cap L=L$ has one connected component.
\item Case $p(\lambda)\neq0$ for all $\lambda\in[0,1)$: the set $\mathcal{Z}\cap L=\emptyset$ has zero connected components.
\item Case $p(\lambda)$ has a finite number of zeros in $[0,1)$: As $p(\lambda)$ is non--negative of degree $2$, it has at most one isolated zero. Thus, in this case $\mathcal{Z}\cap L$ consists of one point and thus has one connected component.
\end{itemize}
In any case, $\mathcal{Z}\cap L$ has at most $1$ connected component.
\end{IEEEproof}
\begin{lemma}\label{lem:discrepancy}
For any $h\geq1$, the cardinal of the set of values of $A_{kl}$ and $B_{kl}$ with entries of the form $\frac{a}{h}+\sqrt{-1}\frac{b}{h}$, $0\leq a,b< h$ such that the mapping in \eqref{eq:linear_mapping} with $C_{kl}=0$ is not surjective is at most
\[
P_h=\left(\frac{2}{h}\sum_{(k,l)\in\Phi}(N_k-d_k)d_l+(M_l-d_l)d_k\right)Q_h,
\]
where $Q_h$ is the total number of $A_{kl},B_{kl}$ with such entries, that is
\[
Q_h=h^{2\sum_{(k,l)\in\Phi}(N_k-d_k)d_l+(M_l-d_l)d_k}
\]
\end{lemma}
\begin{IEEEproof}
From \cite[Th. 3, p. 327]{BlCuShSm98} (note the difference in the notation: our $h$ is $1/h$ in \cite{BlCuShSm98}), we know that
\[
\left|P_h-vol(\mathcal{Z})Q_h\right|\leq \frac{D}{h}\kappa(\mathcal{Z})Q_h,
\]
where $vol(\mathcal{Z})=0$ is the volume (Lebesgue measure) of the proper algebraic variety $\mathcal{Z}$, and $D$ is the (real) dimension of the set of $(A_{kl},B_{kl})$, which is equal to $D=2\sum_{(k,l)\in\Phi}(N_k-d_k)d_l+(M_l-d_l)d_k$. The lemma follows from Lemma \ref{lem:Milnor}.
\end{IEEEproof}
We now prove Theorem \ref{th:Turing}. Let $A_{kl},B_{kl}$ be chosen at random with i.i.d. entries of the form $a+\sqrt{-1}b$, $a,b\in\Z$, $0\leq a,b< h$. Then, the mapping in \eqref{eq:linear_mapping} is surjective if and only if the same mapping but with entries $\frac{a}{h}+\sqrt{-1}\frac{b}{h}$ is surjective, because we are only multiplying each $A_{kl}$ and $B_{kl}$ by $h^{-1}$. From Lemma \ref{lem:discrepancy}, the probability that the linear mapping \eqref{eq:linear_mapping} is not surjective is at most
\[
\frac{P_h}{Q_h}=\frac{2}{h}\sum_{(k,l)\in\Phi}(N_k-d_k)d_l+(M_l-d_l)d_k.
\]
By choosing
\[
h=8\sum_{(k,l)\in\Phi}(N_k-d_k)d_l+(M_l-d_l)d_k,
\]
we guarantee that with probability at least $3/4$ the answer of the algorithm is {\em feasible}.
As already mentioned, one can repeat the test $k$ times to get the probability of having a wrong answer decreasing as $1/4^k$. Note that the integers defining the mapping (\ref{eq:4}) are of bit length bounded above by $1+\log_2 h$, a quantity which is logarithmic in $\sharp(\Phi)$ and $d_j,M_j,N_j$. Hence, the exact arithmetic test can be carried out in time which is polynomial in the same quantities.

%---------------------------------------------------------------

\section{Discussion and Computer Simulations}
\label{sec:simulations}
\subsection{Some results for arbitrary interference channels}
In this subsection, we first show that the proposed feasibility test provides consistent results in agreement with those found in the literature. Moreover, we also discuss scenarios for which the existing DoF outer bounds are not tight. The feasibility test has been evaluated on a vast amount of scenarios, including those covered in \cite{BreslerCartwrightTseToappear} and \cite{Razaviyayn1}, and since its results have always been consistent with all previously known results, here we only show a selection of the most representative cases. Some additional examples can be found in \cite{Gonzalez2012}.\footnote{The reader is invited to test the feasibility of an arbitrary alignment problem at \url{http://gtas.unican.es/IAtest} where the Matlab source code for the floating point test is also provided.}

\begin{example_scenario}
First, consider the simple $(3\times 3,2)^2$ system, which has been already studied in \cite{Yetis10}. Although this system is proper, it is infeasible since it does not satisfy the 2-user outer bound given by \eqref{eq:DoF}. Our test also shows that this system is infeasible.
\end{example_scenario}

%\begin{example_scenario}
%Consider the symmetric $(5\times11,4)^3$ interference network, which was studied in \cite{Dimakis2010_RCRM} and is also proper. This scenario is clearly infeasible (in agreement with our test) because it does not satisfy the outer bound \eqref{eq:DoF_K}, which establishes that the maximum total number of DoF for this network cannot be larger than $11$. By shutting off one beam of the first user, the system $(5\times11,3)(5\times11,4)^2$ could in principle be feasible because it satisfies the mentioned outer bound. Our test shows that this system is actually feasible and thus the outer bound \eqref{eq:DoF_K} is tight for this particular scenario. Furthermore, recent results about the feasibility of the symmetric $3$-user scenario \cite{Bresler2011_3user}, \cite{Jafar2011_3user} establish that the system would be infeasible if $4$ streams per user are transmitted, which is in agreement with the result provided by our test.
%\end{example_scenario}

\begin{example_scenario}
Consider the $3$-user system $\prod_{j=1}^3(7\times13,d_j)$ where the stream distribution among users is not specified. The outer bound \eqref{eq:DoF_K} establishes that total number of DoF cannot exceed $19.5$ in this network, whereas the properness condition in \cite{Yetis10} guarantees that the system is infeasible if more than $5$ DoF per user are transmitted (i.e. a total of $15$ DoF). However, the results in \cite{Bresler2011_3user}, \cite{Jafar2011_3user} provide an even tighter bound which shows that the system is infeasible if $5$ streams per user are transmitted. Our test indicates that the $(7\times13,5)^3$ system is infeasible whereas the system $(7\times13,4)(7\times13,5)^2$ is feasible, which allows us to claim that the maximum total DoF for this network is 14.
\end{example_scenario}

\begin{example_scenario}
The $(4\times4,2)(5\times3,2)(6\times2,2)$ system, which was studied in \cite{Slock10}, satisfies \eqref{eq:DoF} for all 2-user pairs and satisfies all known outer bounds. The proposed test establishes that this system is infeasible.
\end{example_scenario}

\begin{example_scenario}
A controversial example can also be found in \cite{Slock10}: the $(3\times4,2)(1\times3,1)(10\times4,2)$ system. The test proposed in \cite{Slock10} indicates that this system is feasible, while our test establishes that it is infeasible. In our view, the test in \cite{Slock10} gives only necessary (but not sufficient) conditions for feasibility. As our analysis has shown, it is not possible to solve the feasibility problem just by counting variables in all subsets of IA equations, a much more subtle analysis is needed. Similar examples are the $(4\times8,3)^3$ and $(5\times11,4)^3$ networks, which are infeasible according to our test (moreover they violate the outer bound \eqref{eq:DoF_K}) while the test in \cite{Slock10} states they are feasible. We also have numerical evidence that this system is infeasible since iterative algorithms such as \cite{Gomadam11}\cite{Heath09} have not been able to find a solution for this scenario\footnote{Notice, however, that alternating minimization algorithms cannot guarantee convergence to a global minimum, so it cannot be used as a feasibility test.}.
\end{example_scenario}

\begin{example_scenario}
Now, let us consider the $(3\times4,2)(1\times3,1)(10\times4,2)$ system studied in \cite{Yetis10}. It is proper but infeasible, since it violates the $2$-user cooperative outer bound (it is equivalent to the $(4\times7,3)(10\times4,2)$ network). Our test also shows that the system is infeasible.
\end{example_scenario}

\begin{example_scenario}
Consider the $(2\times2,1)^3(3\times5,1)$ system also studied in \cite{Yetis10}. Checking the properness of this scenario involves checking the properness of all the possible subsets of equations. It can be found that the subset of equations which is obtained by shutting down the fourth receiver is improper, therefore the system is infeasible. Our test provides the same result.
\end{example_scenario}

\begin{example_scenario}
A final interesting example is the $(2\times2,1)(5\times5,2)^2(8\times8,4)$ system, which is feasible according to the proposed test. This system has been built by taking the symmetric $(5\times5,2)^4$ system, which is known to be feasible, and transferring 6 antennas from the first user to the fourth. It must noticed that while the total amount of antennas in the network remains constant, the redistribution of antennas has allowed to achieve a total of $9$ DoF instead of the $8$ DoF achieved in the symmetric case. This example gives new evidence for the conjecture settled in \cite{Slock10}, which asserts that for a given total number of DoF, $d_{tot}=\sum_k d_k$, there exist feasible asymmetric MIMO interference systems (that is, with unequal antenna and stream distribution among the links) such that the total number of antennas, $\sum_k(M_k+N_k)$, is less than number of antennas of the smallest symmetric system ($M_k=M$, $N_k=N$, and $d_k=d_{tot}/K$ that can achieve $d_{tot}$.
\end{example_scenario}

Let us finally point out that, in all cases in which our feasibility test was positive, we were able to find an IA solution using the iterative interference leakage minimization algorithm proposed in \cite{Gomadam11}\cite{Heath09}.

\subsection{On the DoF of symmetric $M\times N$ MIMO interference channels}
We have previously shown that the proposed test is in agreement with known results, including those which refer to fully asymmetric systems. Additionally, by using the aforementioned test it is possible to extensively verify conjectures, disprove them or provide additional insights on how the DoF for general interference channels should behave. One such example is the number of linear DoF of the symmetric $K$-user $M\times N$ MIMO interference channel, $(M\times N,d)^K$, which is unknown for $K\geq 4$. For convenience, we use the concept of spatially-normalized degrees of freedom, $d^\star$, introduced in \cite{Jafar2011_3user_full}. When $d^\star$ is an integer, we have an exact DoF characterization. In general it will be a rational number and the actual DoF without spatial extensions can be obtained from it as $d=\lfloor d^\star\rfloor$. To understand the concept of spatially-normalized DoF, let us express $d^\star$ in its rational form $p/q$. Then, scaling the number of antennas by $q$, we have a $qM \times qN$ MIMO interference channel, for which the value $d^\star=p$ is achievable.

We must point out that for the particular case of $K=3$ the linear DoF have been recently obtained \cite{Jafar2011_3user_full,Bresler2011_3user}. In particular, the DoF characterization comprises a piece-wise linear mapping with infinitely many linear intervals over the range of the parameter $\gamma=M/N$ where $M\leq N$ is assumed w.l.o.g. Specifically, the linear DoF are depicted in Fig.~\ref{fig:3user_dof} and are described by the following expression:
\begin{align}\label{eq:compactdof_3user}
d^\star &=&\left\{\begin{array}{ccc}\frac{p}{2p-1}M,&~& \gamma'(p)\leq\frac{M}{N}\leq \gamma(p)\\
\frac{p}{2p+1}N,&~& \gamma(p)\leq\frac{M}{N}\leq
\gamma'(p+1)\end{array} \right. \enskip p\in \mathbb{Z}^+,
\end{align}
where $\gamma'(p)=\frac{p-1}{p}$ and $\gamma(p)=\frac{2p-1}{2p+1}$.

\newlength{\mylength}
\setlength{\mylength}{0.8\textwidth}
\begin{figure*}[p!]
\centering
\includegraphics[width=0.75\textwidth]{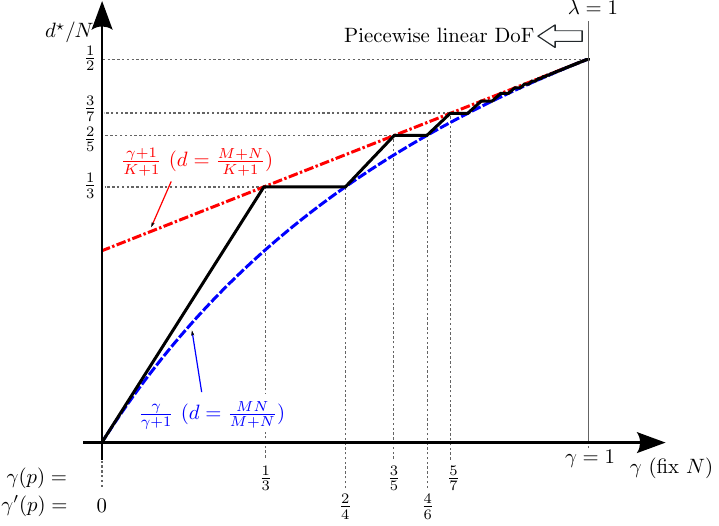}
\caption{Linear degrees of freedom for the $3$-user interference channel as proved in \cite{Jafar2011_3user_full}: $d^{\star}/N$ as a function of $\gamma=M/N$. This figure is included to illustrate the analogy with the results for $K\ge 4$ depicted in Fig. \ref{fig:4user_dof}.}
\label{fig:3user_dof}
\centering
\hbox{\hspace{2.5cm}\includegraphics[width=0.8\textwidth]{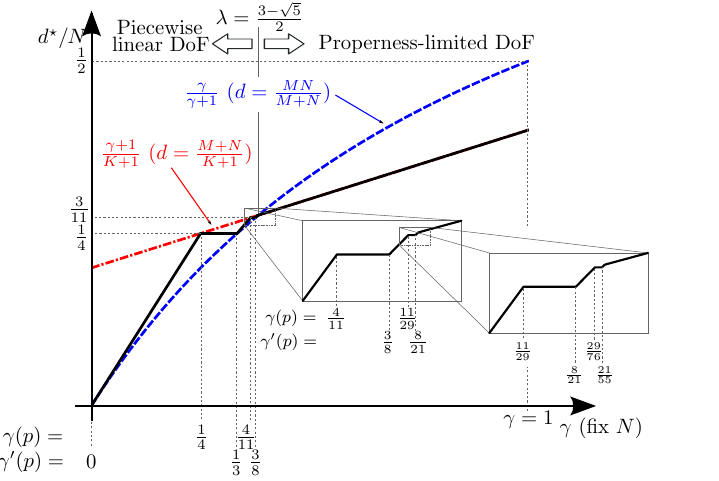}}
\caption{Conjectured linear degrees of freedom for the $4$-user interference channel: $d^{\star}/N$ as a function of $\gamma=M/N$. Similar figures are obtained for all $K$.}
\label{fig:4user_dof}
\end{figure*}

When $K\geq 4$ the exact number of linear DoF is unknown. However, from an information theoretic perspective, and not being restricted to any particular alignment scheme, the DoF have been almost completely characterized by Jafar et al. \cite{Wang_Sun_Jafar_2012} as
\begin{eqnarray}\label{eq:d_IT_Kuser}
d_{IT} &=&\left\{\begin{array}{ccc}M,&~& 0\leq\frac{M}{N}<\frac{1}{K}\\
\frac{N}{K},&~& \frac{1}{K}\leq\frac{M}{N}\leq\frac{1}{K-1}\\
\frac{(K-1)M}{K},&~& \frac{1}{K-1}\leq\frac{M}{N}\leq\frac{K}{K^2-K-1}\\
\frac{(K-1)N}{K^2-K-1},&~& \frac{K}{K^2-K-1}\leq\frac{M}{N}\leq\frac{K-1}{K(K-2)}\\
\frac{MN}{M+N},&~& \frac{K-2}{K^2-3K+1}\leq\frac{M}{N}\leq 1,
\end{array}\right.
\end{eqnarray}
but they are still unknown in the excluded interval, i.e. $\frac{M}{N} \in \left(\frac{K-1}{K(K-2)},\frac{K-1}{K^2-3K+1}\right)$, where they are believed to be $\frac{MN}{M+N}$ as  conjectured in \cite{Wang_Sun_Jafar_2012}. Obviously, the information theoretic DoF is a, sometimes tight, upper bound of the linear DoF without symbol extensions but the extent to which they differ remains unclear.

In order to shed some light on this issue we have extensively executed our test for all the scenarios with $M,N \in [1,100]$ and $K\geq 3$. Our results show two different operating regimes depending on whether $\frac{MN}{M+N}\geq \frac{M+N}{K+1}$ or not. In other words, the regime of operation depends on whether the ratio $\gamma=M/N$ is above or below a threshold value $\lambda=1/2\left(K-1-\sqrt{(K-1)^2-4}\right)$. As an example, Fig.~\ref{fig:4user_dof} shows the linear DoF values per user normalized by $N$ versus the ratio $\gamma=M/N$ for $K=4$. Now, we describe the DoF behaviour for these two regimes in detail for general $K$.

\begin{enumerate}
\item \textit{Regime 1 (Piecewise linear DoF), $\gamma \leq \lambda$:}
\begin{enumerate}
\item We have verified that the linear DoF are given by \eqref{eq:d_IT_Kuser} when
  $0\leq \gamma \leq \frac{(K-1)}{K(K-2)}$ which confirms that in this case the information theoretic DoF can be achieved by linear alignment without symbol extensions.% as claimed in \cite{Wang_Sun_Jafar_2012}.
  \item More interestingly, when $\frac{(K-1)}{K(K-2)} < \gamma \leq \lambda$ we have been able to find several counterexamples that exceed the conjectured value of $\frac{MN}{M+N}$, which was believed to be the information theoretic DoF value. As examples, we enumerate the following feasible systems $(11\times 29,8)^4$, $(44\times 117,32)^4$, $(19\times 71,15)^5$ and $(29\times 139,24)^6$, which clearly exceed the conjectured DoF per user: $7.975$, $31.975$, $14.989$ and $23.994$, respectively. In addition, all systems in this interval seem to follow the same piecewise linear DoF trend described before for the $0\leq \gamma \leq \frac{(K-1)}{K(K-2)}$ region and for the $3$-user interference channel. More precisely, the spatially normalized DoF can be written analogously to \eqref{eq:compactdof_3user} as:
  \par\rlap{\parbox{\columnwidth}
  {
\begin{align}\label{eq:compactdof_Kuser}
d^\star &=&\left\{\begin{array}{ccc}\frac{\gamma(p)+1}{\gamma(p)(K+1)}M,& \gamma'(p)\leq\frac{M}{N}\leq \gamma(p)\\
\frac{\gamma(p)+1}{K+1}N,& \gamma(p)\leq\frac{M}{N}\leq \gamma'(p+1)\end{array} \right. \enskip p\in \mathbb{Z}^+.
\end{align}}
}
where
  \begin{equation}\label{eq:gamma_functions}
\gamma(p)=\frac{\displaystyle \sum_{k=-(p-1)}^{(p-1)}\lambda^k}{\displaystyle \sum_{k=-p}^{p}\lambda^k}\enskip \text{and}\enskip \gamma'(p)=\lambda\:\frac{\displaystyle  \sum_{k=0}^{p-2}\lambda^{2k}}{\displaystyle  \sum_{k=0}^{p-1}\lambda^{2k}}.
\end{equation}
Intuitively, $\gamma(p)$ gives the values of $M/N$ for which there are no antenna redundancies at either side of the link whereas $\gamma'(p)$ gives those for which there is maximum redundancy\footnote{If $\lambda\neq 1$ (i.e. $K\neq 3$) both functions can be simplified: $\gamma(p)=\lambda\frac{1-\lambda^{2p-1}}{1-\lambda^{2p+1}}$ and $\gamma'(p)=\lambda\frac{1-\lambda^{2p-2}}{1-\lambda^{2p}}$.}. Both functions get asymptotically closer as $p$ increases since $\lim_{p\to\infty} \gamma(p)=\lim_{p\to\infty} \gamma'(p)=\lambda$. Specific details on the reasoning leading to \eqref{eq:compactdof_Kuser} are relegated to Appendix \ref{Appendix_conjecture}.
\end{enumerate}
It is worth pointing out that \eqref{eq:compactdof_Kuser} generalizes \eqref{eq:compactdof_3user} and is also consistent with the information theoretic bound in \eqref{eq:d_IT_Kuser}. In fact, for the $3$-user channel, $\lambda$ takes its maximum value, i.e. $\lambda=1$ meaning that the entire $\gamma$ range, $\gamma \in (0,1]$, is covered by this piecewise linear regime as shown in Fig.~\ref{fig:3user_dof}. For $K>3$, the value of $\lambda$ is strictly lower than $1$, approaching to $0$ as $K$ tends to infinity.

\item \textit{Regime 2 (Properness-limited DoF), $\gamma \geq \lambda$:}
For $\gamma$ values above the threshold we have observed that the linear DoF are always given by
\begin{equation}
\label{eq:conjecture_regime2}
d^{\star}=\frac{M+N}{K+1}.
\end{equation}
This means the system is limited by the properness criterion and no proper but infeasible scenarios have been found in this regime.
\end{enumerate}

To sum up, our numerical results lead us to conjecture that the linear DoF of the symmetric $K$-user interference channel ($K\geq3$) are completely characterized by these two regimes thus generalizing the existing results for the $3$-user channel. Formally, it can be written as follows.
\begin{conjecture}
For the $K$-user ($K\geq3$) $M\times N$ MIMO interference channel, the spatially-normalized DoF value per user achievable with linear IA and without time/frequency symbol extensions is given by:

\begin{eqnarray}
\label{eq:conjecture}
d^\star &=&\left\{\begin{array}{ccc} \eqref{eq:compactdof_Kuser}, &~& \frac{M}{N}\leq \lambda\\
\eqref{eq:conjecture_regime2}, &~& \frac{M}{N}\geq \lambda \end{array} \right. \quad,
\end{eqnarray}
where $\lambda=1/2\left(K-1-\sqrt{(K-1)^2-4}\right)$.
\end{conjecture}

It is worth mentioning that during the review process of this paper we have been aware of an independent related work by Liu and Yang \cite{Liu2013} on the degrees of freedom of the symmetric MIMO interference broadcast channel. Their results for the piecewise-limited regime ($\frac{M}{N}\leq \lambda$), although obtained by totally different means, are in perfect agreement with ours. Furthermore, their results when $\frac{M}{N}\geq \lambda$ are based on the test proposed herein and lead them to conjecture, as in \eqref{eq:conjecture}, that the \textit{properness} condition is indeed necessary and sufficient in this regime. This fact still remains unproved.
%----------------------------------------
\section{Conclusions}
\label{sec:conclusions}

This paper gives some new results on the feasibility of interference alignment on the signal space for the K-user MIMO channel with constant coefficients. We use the fact that the input, output and solution variety sets for the IA problem are smooth compact algebraic manifolds. Of particular importance and interest is the study of the projection of the solution variety into its first coordinate and the analysis of their tangent spaces. We prove that for an arbitrary MIMO interference channel IA is feasible iff the algebraic dimension of the solution variety is larger than or equal to the dimension of the input space {\it and} the linear mapping between the tangent spaces of both smooth manifolds given by the first projection is generically surjective, and we provide a simple linear algebra routine, with running time polynomial in the input parameters $d_j,M_j,N_j,\sharp(\Phi)$, to decide if the scenario is feasible. The matrix representing this linear mapping can be easily obtained and the feasibility of IA amounts to checking whether this matrix is full rank or not. Proper but infeasible systems correspond to cases in which the dimension of the solution variety coincides with the dimension of the input space, but the mapping is not surjective, that is, the solution variety is mapped to a zero-measure set of MIMO interference channels. We have evaluated our feasibility test on many examples, some of them served to corroborate known results, others showed the non-tightness of existing DoF outer bounds for this setting or provided evidence on the advantages of unequal antennas and stream distribution for DoF maximization. Additionally, an extensive execution of our test on symmetric scenarios allowed us to establish a conjecture on the DoF of the $K$-user interference channel which generalizes already known results for $K=3$.

%------------------------------------------------------------------

% or

% Use this command to get the appendices' numbers in "A", "B" instead ofthe
% default capitalized Roman numerals ("I", "II", etc.).
% However, the capital letter form may result in awkward subsection numbers
% (such as "A-A"). Capitalized Roman numerals are the default.
%\useRomanappendicesfalse
%
\appendices
\section{Review of some results from algebraic geometry and differential topology}
\label{app:difftop}
A key point of our analysis is a subtle use of the notion of compactness of spaces. We introduce this fundamental mathematical concept in the following lines. Recall that a topological space $X$ is just a set where a collection $\tau\subset\{\text{subsets of }X\}$ of ``open subsets'' has been chosen, satisfying three conditions:
\begin{enumerate}
\item the empty set and the total set $X$ are in $\tau$,
\item the intersection of a finite number of elements in $\tau$ is again in $\tau$, and
\item the union of any collection of elements in $\tau$ is again in $\tau$.
\end{enumerate}
For example, $\R^n$ with the usual definition of ``open set'' is a topological space. Any subset $A\subseteq\R^n$ (for example, a sphere or a linear subspace) then inherits a structure of topological space, with open sets being those obtained by intersecting an open set of $\R^n$ with $A$. More generally, any (smooth) manifold is by definition a topological space and any subset of a manifold inherits a structure of topological space.

A subset $A\subseteq X$ of a topological space is called {\em compact} if the following property holds: given any collection of open sets of $X$ such that their union contains $A$, there exist a finite subcollection which also contains $A$. This is not a particularly intuitive definition, but it permits to obtain many results, notoriously a fundamental result due to Ehressman that will be recalled below. From the Heine--Borel Theorem, a subset of $\R^n$ or $\C^n$ is compact if and only if it is closed (in the usual definition) and bounded. Thus, the sphere is compact but a linear subspace is not.

Using the definition, note that a given manifold $X$ is itself compact if any collection of open subsets whose union is $X$ has a finite subcollection that covers $X$. For example, $\R^n$ is not compact (the union for $m\geq1$ of open balls of radius $m$ covers $\R^n$ but no finite subcollection of these balls covers $\R^n$). It is not obvious but it is true that the projective spaces $\P(\R^n)$ and $\P(\C^n)$ are both compact. We will finally use the following basic fact: if $X$ is compact and $A\subseteq X$ is closed, then $A$ is compact as well.

We will also use some basic notions related to regular mappings: let $\varphi:X\ra Y$ be a smooth mapping where $X$ and $Y$ are smooth manifolds. For every $x\in X$, the derivative is a linear mapping between the tangent spaces, $D\varphi(x):T_xX\ra T_{\varphi(x)}Y$. A {\em regular point} of $\varphi$ is a point such that $D\varphi(x)$ is surjective (which requires $\dim(X)\geq\dim(Y))$. A {\em critical point} is a $x\in X$ which is not regular. Similarly, a {\em regular value} of $\varphi$ is an element $y\in Y$ such that for every $x\in X$ such that $\varphi(x)=y$, $x$ is a regular point. That is, $y\in Y$ is a regular value if every point mapped to $y$ is a regular point. This includes, by convention, the case $\varphi^{-1}(y)=\emptyset$. If $y$ is not a regular value, we say that it is a {\em critical value}. Note that
\[
\varphi\{\text{critical points of }\varphi\}=\{\text{critical values of }\varphi\}.
\]
If $x$ is a regular point of $\varphi$ we say that $\varphi$ is a {\em submersion} at $x$. If $\varphi$ is a submersion at every point (equivalently, every $x\in X$ is a regular point of $\varphi$) then we simply say that $\varphi$ is a submersion.

%A related concept is that of transversality: two submanifolds $Y,Z$ of a manifold $X$ are {\em transversal} if for every $x\in Y\cap Z$ we have $T_xY+T_xZ=T_xX$, namely if the tangent spaces of $Y$ and $Z$ span that of $X$.

We now recall a few results from regular mappings; the reader may find them for example in \cite[Ch. 1]{GuilleminPollack1974} or \cite{Schwartz1968}:
%\begin{theorem}[Local Submersion Theorem]
%If $\varphi:X\ra Y$ is a submersion at $x$ then $f$ is locally equivalent to an affine projection. In particular, there is an open set $U\subseteq X$ containing $x$ such that $\varphi(U)\subseteq Y$ is an open set.
%\end{theorem}
\begin{theorem}[Preimage Theorem]
If $Y_0\subseteq Y$ is a submanifold such that every $y\in Y_0$ is a regular value of $\varphi:X\ra Y$ then $Z=\varphi^{-1}(Y_0)$ is a submanifold of $X$ of dimension $\dim (Z)=\dim(X)-\dim(Y)+\dim(Y_0)$. Moreover, the tangent space $T_xZ$ at $x$ to $Z$ is the kernel of the derivative $D\varphi(x):T_xX\ra T_yY$.
\end{theorem}
%\begin{theorem}[Intersection of transversal manifolds]
%The intersection $Y\cap Z$ of two transversal submanifolds $Y,Z$ of $X$ is again a submanifold. Its codimension is the sum of the codimensions of $Y$ and $Z$, and its tangent space at $x$ is the intersection of the tangent spaces at $x$ of $Y$ and $Z$.
%\end{theorem}
\begin{theorem}[Sard's Theorem]
If $X$ and $Y$ are manifolds and $\varphi:X\ra Y$ is a smooth mapping, then almost every point of $Y$ is a regular value of $\varphi$.
\end{theorem}
%A version of Sard's Theorem for algebraic varieties is usually called Bertini's Theorem, see for example \cite{Sh94a}:
%\begin{theorem}[Bertini's Theorem for projections]
%Let $X,Y$ be smooth projective algebraic subvarieties of $\P(\C^a)$ and $\P(\C^b)$ respectively. Let $Z\subseteq X\times Y$ be a smooth algebraic subvariety and let $\pi:Z\ra X$ be the projection onto the first coordinate. Then, the set of critical values of $\pi$ is a proper algebraic subset of $X$.
%\end{theorem}
\begin{remark}
Note that it can happen that every $x\in X$ is a critical point: this simply means that every $y\in\varphi(X)$ is a critical value, which by Sard's theorem means that $\varphi(X)$ has zero--measure in $Y$. This phenomenon is behind case 1 of Theorem \ref{th:distinct}.
\end{remark}

Another tool that we will use is a celebrated theorem by Ehresmann, a foundational result in differential topology. Before writing it, we recall that a {\em fiber bundle} is a tuple $(E,B,\pi,F)$ where $E,B,F$ are manifolds and $\pi:E\ra B$ is a continuous surjective mapping that is locally like a projection $B\times F\ra E$, in the sense that for any $x\in E$ there exists an open neighborhood $U\subseteq B$ of $\pi(x)$ such that $\pi^{-1}(U)$ is homeomorphic to the product space $U\times F$. For example, $E=\R^2\setminus\{0\}$ is a fiber bundle with base space $B$ the unit circle and fiber $F=\R$, because locally $\R^2\setminus\{0\}$ is as a product space of a short piece of the circle and a line (which goes from $0$ to $\infty$ with no extremes). Fiber bundles are very useful objects in the study of geometry and they are closely related to regular values as the following result (see \cite{Ehresmann1951} or \cite[Th. 5.1]{Rabier1997} for a more general version) shows:
\begin{theorem}[Ehresmann's Theorem]
Let $X,Y$ be smooth manifolds with $Y$ connected, let $U\subseteq X$ be a nonempty open subset of $X$ and let $\pi:U\rightarrow Y$ satisfy:
\begin{itemize}
\item $\pi$ is a submersion, and
\item $\pi$ is proper, i.e. the inverse image of a compact set is a compact set.
\end{itemize}
Then, $\pi:X\ra Y$ is a fiber bundle, and $\pi(U)=Y$.
\end{theorem}
In the precedent theorem, if $X$ is compact and $\dim(X)=\dim(Y)$, then the inverse image of any point is a finite set and the fact that every point is regular with the Inverse Mapping Theorem implies that $\pi$ is actually a covering map, that is every point $y\in Y$ has an open  neighborhood $V$ whose preimage by $\pi$ which is equal to a finite number of open sets of $X$, each of them homeomorphic to $V$. Thus:
\begin{corollary}\label{cor:covering}
If in Ehresmann's Theorem we assume moreover that $X$ is compact and $\dim(X)=\dim(Y)$ then $\pi$ defines a covering map. In particular, this implies that every $y\in Y$ has a finite number of preimages, and that number is the same for all $y\in Y$.
\end{corollary}

We recall also some known facts from algebraic geometry. Our basic references are \cite{Sh94a,Mu76}. Given complex vector spaces $V_1,\ldots,V_l$, the {\em Segre embedding} is a mapping from the product of projective spaces $\P(V_1)\times\cdots\times\P(V_l)$ into a higher dimensional projective space $\P(\mathcal{T})$ (where $\mathcal {T}$ is a high--dimensional vector space) such that:
\begin{itemize}
\item it is a diffeomorphism into its image (more specifically, it is an embedding), and
\item the image of an algebraic subvariety is an algebraic subvariety and viceversa.
\end{itemize}
The Segre embedding is useful because it allows us to treat some objects (for example, products of Grassmannians) as algebraic subvarieties of a high--dimensional projective space. We will use this at some point combined with the following result
\begin{theorem}[Main Theorem of Elimination Theory]
Let $Z\subseteq \P(\C^a)\times\P(\C^b)$ be an algebraic variety. Then,
\[
\pi_1(Z)=\{x\in \P(\C^a):\exists\,y\in \P(\C^b),(x,y)\in Z\}
\]
is an algebraic subvariety of $X$.
\end{theorem}
%--------------------------------------------------------------
%
% % you can choose not to have a title for an appendix
% % if you want by leaving the argument blank
%------------------------------------------------------------
\section{Proof of Lemma \ref{lem:particularH}}
\label{Appendix2}
\begin{IEEEproof}
Let $(U,V), (A,B) \in\mathcal{S}$ be two points, and assume that we have chosen affine representatives that we denote by the same letters ${U},{V},{A},{B}$. Note that there exist nonsingular square matrices $Q_j$ of size $N_j$ and $P_j$ of size $M_j$ such that ${U}_j=Q_j{A}_j$ and ${V}_j=P_j{B}_j$. Consider the following mapping
\[
\begin{matrix}
\pi_2^{-1}(U,V)&\ra&\pi_2^{-1}(A,B)\\
H_{kl}&\mapsto& Q_k^TH_{kl}P_l
\end{matrix}
\]
which is a linear bijection. Thus, $\pi_2^{-1}(U,V)$ is empty or nonempty for every $(U,V)\in\mathcal{S}$ and it suffices to prove the claim for {\em some} $(U,V)\in\mathcal{S}$. If it is nonempty for some (thus, all) $(U,V)$, let $(U,V)\in\mathcal{S}$ be a regular value of $\pi_2$. Then, from the Preimage Theorem $\pi_2^{-1}(U,V)$ is a smooth submanifold of $\CV$ of the claimed dimension (the dimension of $\CV$ is given in Lemma \ref{lem:V}.) Moreover, it is given by the nullset of a set of linear (in $H$) equations and is thus a product of projective vector subspaces as claimed.

We now discard the case that $\pi_2^{-1}(U,V)$ is empty for every $(U,V)\in\mathcal{S}$ (equivalently, $\CV$ is empty).
Note that since we have assumed \eqref{eq:p2p_assumption2} holds, the particularly simple element $(H,U,V)$, first described in Section \ref{sec:floatingpoint_test}, is in $\CV$ and hence $\CV\neq\emptyset$.
\end{IEEEproof}

\section{Proof of Lemma \ref{lem:aux}}
\label{Appendix6}

\begin{IEEEproof}
Let $\theta_1$ be the mapping of (\ref{eq:4}) for representatives $({H}_1,{U}_1,{V}_1)$ of $(H,U,V)$, and similarly let $\theta_2$ be the mapping of (\ref{eq:4}) for representatives $({H}_2,{U}_2,{V}_2)$ of $(H,U,V)$. We need to prove that if $\theta_1$ is surjective then so is $\theta_2$. Because both affine points are representatives of the same $(H,U,V)$, there exist complex numbers $(\lambda_{kl})_{(k,l)\in\Phi}$ and nonsingular matrices $Q_k\in\mathbb{C}^{d_k\times d_k}$, $k\in\Phi_R$, and $P_l\in\mathbb{C}^{d_l\times d_l}$, $l\in\Phi_T$, such that
\[
({H}_2)_{kl}=\lambda_{kl}({H}_2)_{kl},\quad ({U}_2)_k=({U}_1)_k Q_k,\quad ({V}_2)_l=({V}_1)_l P_l.
\]
Let $\dot R=(\dot R_{kl})_{(k,l)\in\Phi}\in\prod_{(k,l)\in\Phi}\mathbb{C}^{d_k\times d_l}$. If $\theta_1$ is surjective, there exist $(\{\dot{{U}}_k\}_{k\in\Phi_R},\{\dot{{V}}_l\}_{l\in\Phi_T})$ such that
\[
\dot{{U}}_k^T ({H}_1)_{kl}{({V}_1)}_l+{({U}_1)}_k^T({H}_1)_{kl}\dot{{V}}_l=\lambda_{kl}^{-1}(Q_k^T)^{-1}\dot R_{kl}P_l^{-1}.
\]
Then,
\begin{equation*}
\begin{split}
(\theta_2&(\{\dot{{U}}_kQ_k\}_{k\in\Phi_R},\{\dot{{V}}_lP_l\}_{l\in\Phi_T}))_{kl}\\
&=Q_k^T\dot{{U}}_k^T ({H}_2)_{kl}{({V}_2)}_l+{({U}_2)}_k^T({H}_2)_{kl}\dot{{V}}_lP_l\\
&=\lambda_{kl}\left(Q_K^T\dot{{U}}_k^T ({H}_1)_{kl}({V}_1)_l P_l+Q_k^T({U}_1)_k^T ({H}_1)_{kl}\dot{{V}}_lP_l\right)\\
&=\lambda_{kl}Q_k^T\left(\dot{{U}}_k^T ({H}_1)_{kl}({V}_1)_l +({U}_1)_k^T ({H}_1)_{kl}\dot{{V}}_l\right)P_l\\
&=\lambda_{kl}Q_k^T\left(\lambda_{kl}^{-1}(Q_k^T)^{-1}\dot R_{kl}P_l^{-1}\right)P_l=\dot R_{kl}.
\end{split}
\end{equation*}
Thus, $\theta_2$ is surjective as claimed.
\end{IEEEproof}

%------------------------------------------------------------
\section{Proof of Proposition \ref{prop:charsigma}}
\label{Appendix7}

\begin{IEEEproof}
Assume first that $\theta$ is surjective, and let $(H,U,V)$ be some fixed affine representatives. For any tangent vector $\dot{H}$, let $\dot R=(\dot R_{kl})_{(k,l)\in\Phi}\in \prod_{(k,l)\in\Phi}\mathbb{C}^{d_k\times d_l}$ be defined as
\[
\dot R_{kl}=-{U}_k^T\dot{{H}}_{kl}{V}_l.
\]
Because $\theta$ is surjective, there exists $(\dot{{U}},\dot{{V}})\in\theta^{-1}(\dot R)$, that is $(\dot{{U}},\dot{{V}})$ satisfying
\begin{equation}\label{eq:ecuacion}
\dot{{U}}_k^T {H}_{kl}{{V}}_l+{{U}}_k^T{H}_{kl}\dot{{V}}_l=-{U}_k^T \dot{{H}}_{kl}{V}_l,\quad (k,l)\in\Phi.
\end{equation}
%We note that we can choose $\dot{{V}}$ in such a way that all the columns of $\dot{{V}}_l$ are orthogonal to the column space of ${V}_l$. Indeed, for any $\dot{{W}}_l$ write $\dot{{W}}_l=\dot{{V}}_l+A$, where the columns of $A$ are spanned by the columns of ${V}_l$. Then, from $U_k^TH_{kl}V_l=0$ we have that ${{U}}_k^T{H}_{kl}\dot{{W}}_l= {{U}}_k^T{H}_{kl}\dot{{V}}_l$. Thus, we can take $\dot{{V}}_l$ with that property, and a similar argument shows that we can choose $\dot{{U}}_k$ with the equivalent property too. Now, this precisely means that $(\dot{{U}},\dot{{V}})\in T_{({U},{V})}{\mathcal{S}}$ is in the tangent space to ${\mathcal{S}}$, and f
Note that the equations defining $\CV$ are precisely $U_k^TH_{kl}V_l=0$, $(k,l)\in\Phi$, and thus from the Preimage theorem we can cover the tangent space to $\CV$ at $(H,U,V)$ with those $(\dot H,\dot U,\dot V)$ satisfying \eqref{eq:ecuacion}. We conclude that $(\dot{{H}},\dot{{U}},\dot{{V}})$ is in the tangent space to $\CV$ at $({H},{U},{V})$, and thus $D\pi_1(H,U,V)(\dot H,\dot U,\dot V)=\dot H$, which means that $D\pi_1(H,U,V)^{-1}(\dot H)\neq\emptyset$. As $\dot H$ was chosen generically, we conclude that $\pi_1$ is a submersion at $(H,U,V)$, namely $(H,U,V)$ is a regular point of $\pi_1$ as wanted. This finishes the ``if'' part of the proposition.

The ``only if'' part is a converse reasoning: assume that $(H,U,V)$ is a regular point of $\pi_1$. This means that for every $\dot H\in T_H\mathcal{H}$ there exist $(\dot U,\dot V)\in T_{(U,V)}\mathcal{S}$ such that $(\dot H,\dot U,\dot V)\in T_{(H,U,V)}\CV$, which means that these tangent vectors satisfy \eqref{eq:ecuacion}. Let $(\dot{R}_{kl})_{(k,l)\in\Phi} \in\prod_{(k,l)\in\Phi}\mathbb{C}^{d_k\times d_l}$. Now, because $U_k$ and $V_l$ are representatives of an element of the Grassmanian, they are full rank and thus we can write $\dot{R}_{kl}=-U_k^T\dot{{H}}_{kl} V_l$ for some $\dot H_{kl}$. Then, \eqref{eq:ecuacion} reads
\[
\dot{{U}}_k^T {H}_{kl}{{V}}_l+{{U}}_k^T{H}_{kl}\dot{{V}}_l=-{U}_k^T \dot{{H}}_{kl}{V}_l=\dot R_{kl},\quad (k,l)\in\Phi,
\]
that is all such $\dot R_{kl}$ have a preimage by $\theta$, and $\theta$ is surjective.
% Identifying as usual the tangent space to the projective space with the orthogonal complement of a representative we can identify $T_H\mathcal{H}$ with $\prod_{(k,l)\in\Phi}({H}_{kl})^\perp$, and then the fact that $\pi_1$ is a submersion at $(H,U,V)$ implies that there exists a tangent vector $(\dot {{H}},\dot{{U}},\dot{{V}})\in T_{({H},{U},{V})}\CV$, which from Lemma \ref{lem:Vhat} means
%\begin{equation}\label{eq:7}
%\dot{{U}}_k^T {H}_{kl}{{V}}_l+{{U}}_k^T{H}_{kl}\dot{{V}}_l=-{U}_k^T \dot{{H}}_{kl}{V}_l=\dot R_{kl},\quad (k,l)\in\Phi,
%\end{equation}
%that is all such $\dot R_{kl}$ have a preimage by $\theta$. Now, we claim that
%\[
%\{ U_k^T\dot{{H}}_{kl} {V}_l:\dot{{H}}_{kl}\perp {H}_{kl}\}=\mathbb{C}^{d_k\times d_l},
%\]
%which will prove the proposition. To see this last equality, note that, because ${U}_k$ and ${V}_l$ are of maximal rank, for any $\dot R_{kl}$ there exists $\dot{{H}}_{kl}$ such that $ U_k^T\dot{{H}}_{kl} {V}_l=\dot R_{kl}$. Now, $ U_k^T{{H}}_{kl} {V}_l=0$ and hence we can choose $\dot{{H}}_{kl}$ to be perpendicular to ${{H}}_{kl}$ and we are done.
\end{IEEEproof}
%-------------------------------------------------
%------------------------------------------------------------
\section{Proof of Proposition \ref{prop:sigma}}
\label{Appendix8}

\begin{IEEEproof}
From Proposition \ref{prop:charsigma}, $\Sigma'$ can be written as the set of $(H,U,V)$ such that all the minors of the matrix defining $\theta$ are equal to $0$. Thus, $\Sigma'$ is an algebraic subvariety of $\CV$.
The set $\mathcal{H}$ is a product of projective spaces and hence the associated Segre embedding defines a natural embedding
\begin{equation}\label{eq:phi2}
\varphi_1:\mathcal{H}\rightarrow\P(\mathcal{T}_1),
\end{equation}
where $\mathcal{T}_1$ is a high--dimensional complex vector space.

Let $\wedge^a(\C^b)$ the $a$--th exterior power of $\C^b$. Then, the Grassmannian $\G{a}{b}$ can be seen as an algebraic subset of a complex projective space $\P(\wedge^a(\C^b))$, and as a compact complex manifold of (complex) dimension $a(b-a)$  (see for example \cite[p.42]{Sh94a} and \cite[p. 175--176]{Whitney1972}). The Segre embedding defines a natural embedding
\begin{equation}\label{eq:phi1}
\varphi_2:\mathcal{S}\rightarrow\P(\mathcal{T}_2),
\end{equation}
where $\mathcal{T}_2$ is a certain (high--dimensional) complex vector space. Both $\varphi_1$ and $\varphi_2$ define diffeomorphisms between their domains and ranges, as does the product mapping $\varphi_1\times\varphi_2$, and they preserve algebraic varieties in both ways. We can thus identify $\mathcal{H}\equiv\varphi_1(\mathcal{H})$, $\mathcal{S}\equiv\varphi_2(\mathcal{S})$ and see $\CV$ as an algebraic subvariety of the product space
\[
\CV\equiv (\varphi_1\times\varphi_2)(\CV)\subseteq\P(\mathcal{T}_1)\times\P(\mathcal{T}_2).
\]
The Main Theorem of Elimination Theory then grants that $\Sigma=\pi_1(\Sigma')$ is an algebraic subvariety of $\mathcal{H}$. We moreover have that it is a proper subvariety because by Sard's Theorem it has zero--measure in $\mathcal{H}$.
%where $\equiv$ means that we identify a point $(H,U,V)\in\CV$ with its image $(\varphi_1\times\varphi_2)(H,U,V)=(\varphi_1(H),\varphi_2(u,v))$, and that this identification preserves algebraic varieties in both ways. Note also that the following is a commutative diagram
%\begin{equation}\label{eq:comm}
%\begin{matrix}
%\CV&\overset{\varphi_1\times\varphi_2}{\ra}&(\varphi_1\times\varphi_2)(\CV)&\subseteq&\P(\mathcal{T}_1)&\times&\P(\mathcal{T}_2)\\
%\pi_1\downarrow&&\downarrow\pi&&&&\\
%\mathcal{H}&\underset{\varphi_1}{\ra}&\varphi_1(\mathcal{H})&\subseteq&\P(\mathcal{T}_1)&&
%\end{matrix}
%\end{equation}
%where $\pi$ is simply the projection onto the first coordinate. Now, $\varphi_1\times\varphi_2$ is a diffeomorphism onto its image and hence critical points and values of $\pi_1$ correspond to critical points and values of $\pi$. Let $\mathcal{C}\subseteq\varphi_1(\mathcal{H})$ be the set of critical values of $\pi$. From Bertini's Theorem, $\mathcal{C}$ is a proper algebraic subset of $\varphi_1(\mathcal{H})$ and hence so is
%\[
%\Sigma=\varphi_1^{-1}(\mathcal{C})=\{H\in\mathcal{H}:H\text{ is a critical value of }\pi_1\}\subseteq\mathcal{H},
%\]
%as desired.
\end{IEEEproof}

%--------------------------------------------------------------------------------------
\section{Derivation of \eqref{eq:gamma_functions}}
\label{Appendix_conjecture}
The execution of the proposed test for a large number of scenarios suggests that $\gamma(p)$ and $\gamma'(p)$, which we will indistinctly denote as $\gamma^\star(p)$, are given by
\begin{eqnarray}\label{eq:gamma_recursion}
\gamma^\star(p)=\frac{F^\star_{p}}{F^\star_{p+1}}\notag
\end{eqnarray}
where $F^\star_p$ satisfies the recurrence relation $F^\star_{p+1}=(K-1)F^\star_{p}-F^\star_{p-1}$ with initial conditions $F_1=1,\ F_0=-1$ (for $\gamma(p)$), and $F'_1=0,\ F'_0=-1$ (for $\gamma'(p)$). Sequences satisfying this recurrence equation are known as Lucas Sequences because any such a sequence can be represented as a linear combination of the Lucas sequences of first and second kind. Lucas sequences are a generalization of other famous sequences including Fibonacci numbers, Mersenne numbers, Pell numbers, Lucas numbers, etc. The interested reader can find a good introduction to Lucas sequences in \cite[Chapter 17]{Dickson2012}.

For convenience, we rewrite the recurrence relation in matrix form $f^\star_p=Af^\star_{p-1}$, where
\begin{equation}
\label{eq:matrix_recursion}
\underbrace{\left(\begin{array}{c}
F^\star_{p+1}\\
F^\star_{p}
\end{array}\right)}_{f^\star_p}=\underbrace{\left(\begin{array}{cc}
(K-1) & -1\\
1 & 0
\end{array}\right)}_{A}\underbrace{\left(\begin{array}{c}
F^\star_{p}\\
F^\star_{p-1}
\end{array}\right)}_{f^\star_{p-1}}.\notag
\end{equation}
Now, we are interested in writing $f^\star_p$ as a function of the initial conditions, i.e. $f^\star_p=A^pf^\star_0$. In order to do so, we first need the eigenvalue decomposition of $A$. The eigenvalues are the roots of the characteristic polynomial
\begin{equation}
\label{eq:characteristic_polynomial}
\operatorname{det}(A-\lambda I)=\lambda^2-(K-1)\lambda+1=0,\notag
\end{equation}
which are given by
\begin{equation}
\label{eq:roots_characteristic_polynomial}
\lambda_\pm=\frac{1}{2}((K-1)\pm\sqrt{(K-1)^2-4}).\notag
\end{equation}
Notice that given $\operatorname{det}(A)=1$, $\lambda_-=1/\lambda_+$. Thus, for convenience we define $\lambda=\lambda_-$ and factorize $A^p=S\Lambda^p S^{-1}$:
\begin{equation}
\label{eq:spectral_decomposition}
A^p=\underbrace{\left(\begin{array}{cc}
1/\lambda &\lambda\\
1& 1
\end{array}\right)}_{S}
\underbrace{
\left(\begin{array}{cc}
1/\lambda^p &0\\
0& \lambda^p
\end{array}\right)}_{\Lambda^p}
\underbrace{\left(\begin{array}{cc}
1 & -\lambda\\
-1& 1/\lambda
\end{array}\right)\frac{\lambda}{1-\lambda^2}}_{S^{-1}},\notag
\end{equation}
where the columns of $S$ are the eigenvectors of $A$. Then, using the fact that $f^\star_p=S\Lambda^p S^{-1}f^\star_0$, it is straightforward to obtain a compact expression for $F^\star_p$:
\begin{equation}
\label{eq:F_p}
F^\star_p=\lambda^{-p+1}\left(F^\star_1 \sum_{k=0}^{p-1}\lambda^{2k}-F^\star_0 \sum_{k=0}^{p-2}\lambda^{2k+1}\right).
\end{equation}
Finally, when the corresponding initial conditions are substituted in \eqref{eq:F_p}, we can write
  \begin{equation*}
\gamma(p)=\frac{F_p}{F_{p+1}}=\frac{\displaystyle \sum_{k=-(p-1)}^{(p-1)}\lambda^k}{\displaystyle \sum_{k=-p}^{p}\lambda^k}
\end{equation*}
and
\begin{equation*}
\label{eq:gamma_functions_appendix}
\gamma'(p)=\frac{F'_p}{F'_{p+1}}=\lambda \quad\frac{\displaystyle  \sum_{k=0}^{p-2}\lambda^{2k}}{\displaystyle  \sum_{k=0}^{p-1}\lambda^{2k}}.
\end{equation*}
A final observation is that
$
\lim_{p\to\infty} \gamma(p) =\lim_{p\to\infty} \gamma'(p) = \lim_{p\to\infty}\frac{F^\star_p}{F^\star_{p+1}} =\lambda
$
and, thus, $\lambda$ is also a threshold value separating the so-called piecewise linear and properness-limited DoF regimes.

\section*{Acknowledgment}
We would like to thank anonymous reviewers and the Associate Editor for helpful comments and suggestions which significantly improved the quality of the paper.
% % optional entry into table of contents (if used)
% %\addcontentsline{toc}{section}{Acknowledgment}



% http://www.ctan.org/tex-archive/macros/latex/contrib/supported/IEEEtran/bibtex
%\bibliographystyle{IEEEtran.bst}
% H.~Kopka and P.~W. Daly, \emph{A Guide to {\LaTeX}}, 3rd~ed.\hskip 1em
% plus
% \end{thebibliography}

\bibliographystyle{IEEEtran}
\bibliography{myreferences}

%\begin{biography}[{\includegraphics[width=1in,height=1.25in,clip,keepaspectratio]{mshell}}]{Michael Shell}
% where an .eps filename suffix will be assumed under latex, and a .pdf suffix
% will be assumed for pdflatex; or if you just want to reserve a space for
% a photo:

\begin{IEEEbiographynophoto}{\'Oscar~Gonz\'alez} (S’10) received the B.S. degree in
Telecommunication Engineering from the University
of Cantabria, Santander, Spain, in 2009, where he
has been working toward the Ph.D. degree in the
Communications Engineering Department, under the
supervision of I. Santamaría, since 2009. During 2012, he was a visiting researcher at the Wireless Networking and Communications Group (The University of Texas at Austin). His current research interests include signal processing for wireless communications, interference management/alignment techniques, multiple-input multiple-output (MIMO) systems and the development of wireless communications demonstrators. He has been involved in several national and international research
projects on these topics.
\end{IEEEbiographynophoto}

\begin{IEEEbiographynophoto}{Carlos~Beltr\'an}
received the Ph. D Degree in Mathematics from the Universidad de Cantabria, Spain, in 2006. He held a postdoctoral fellowship at the U. of Toronto during 2007 and 2008, and is currently a Profesor Titular at the Universidad de Cantabria. He has visited for short periods the \'{E}chole Polytechnique (Paris, 2004), the Universit\'{e} Paul Sabatier (Toulouse, 2005) and the Instituto Nacional de Matem\'{a}tica Pura e Aplicada (Rio de Janeiro, 2008). He was awarded by the Real Sociedad Matem\'{a}tica Espa\~{n}ola the Jose Luis Rubio de Francia 2010 prize for his solution to Smale's 17th problem. His research interests include Numerical Analysis, Complexity and Numerical Algebraic Geometry, as well as applied problems.
\end{IEEEbiographynophoto}

\begin{IEEEbiographynophoto}{Ignacio~Santamar\'ia}
(M’96, SM’05) received his Telecommunication Engineer Degree and his Ph.D. in electrical engineering from the Universidad Polit\'{e}cnica de Madrid (UPM), Spain, in 1991 and 1995, respectively. In 1992 he joined the Department of Communications Engineering, University of Cantabria, Spain, where he is currently Full Professor. He has co-authored more than 150 publications in refereed journals and international conference papers and holds 2 patents. His current research interests include signal processing algorithms and information-theoretic aspects of multi-user multi-antenna wireless communication systems, multivariate statistical techniques and machine learning theories. He has been involved in numerous national and international research projects on these topics. He has been a visiting researcher at the Computational NeuroEngineering Laboratory (University of Florida), and at the Wireless Networking and Communications Group (The University of Texas at Austin).
Dr. Santamaria was a Technical Co-Chair of the 2nd International ICST Conference on Mobile Lightweight Wireless Systems (MOBILIGHT 2010), Special Sessions Co-Chair of the 2011 European Signal Processing Conference (EUSIPCO 2011), and General Co-Chair of the 2012 IEEE Workshop on Machine Learning for Signal Processing (MLSP 2012). Since 2009 he has been a member of the IEEE Machine Learning for Signal Processing Technical Committee. Currently, he serves as Associate Editor of the IEEE Transactions on Signal Processing. He was a co-recipient of the 2008 EEEfCOM Innovation Award, as well as coauthor of a paper that received the 2012 IEEE Signal Processing Society Young Author Best Paper Award.

\end{IEEEbiographynophoto}

% \begin{biographynophoto}{John Doe}
% \end{biographynophoto}
%
% % insert where needed to balance the two columns on the last page
% %\newpage
%
% \begin{biographynophoto}{Jane Doe}
% \end{biographynophoto}


%\vfill


\end{document}